\newif\ifreport\reporttrue
\documentclass[12pt,draftclsnofoot,onecolumn]{IEEEtran}
\usepackage{amsmath,amsfonts,amssymb,bm}
\usepackage[tight,footnotesize]{subfigure}
\usepackage{graphicx,stfloats}
\usepackage{algorithm,algpseudocode}
\usepackage{cite,url}
\usepackage{color}
\usepackage{theorem} 

\allowdisplaybreaks[4]  

\newtheorem{Proposition}{Proposition}
\newtheorem{Lemma}{Lemma}

\theorembodyfont{\rmfamily}

\newcommand{\In}{\,{\in}\,}
\newcommand{\st}{\ensuremath{\text{s.t.}}}
\newcommand{\Iset}{\ensuremath{\mathcal{I}}}
\newcommand{\Nset}{\ensuremath{\mathcal{N}}}

\newcommand{\Rset}{\ensuremath{\mathbb{R}}}
\IEEEoverridecommandlockouts    
\begin{document}

\title{Capacity Region Bounds and Resource Allocation for Two-Way OFDM Relay Channels}

\linespread{1.2}\selectfont
\author{Fei~He$^\star$, Yin~Sun$^\dag$, Limin~Xiao$^\star$, Xiang~Chen$^\star$, Chong-Yung~Chi$^\ddagger$, and Shidong~Zhou$^\star$
\thanks{This work is supported by National Basic Research Program of China (2012CB316002), National S\&T Major Project (2010ZX03005-003), National NSF of China (60832008), China's 863 Project (2009AA011501), Tsinghua Research Funding (2010THZ02-3), International Science Technology Cooperation Program (2010DFB10410), National Science Council, Taiwan, under Grant NSC-99-2221-E-007-052-MY3, Ericsson Company and the MediaTek Fellowship. The material in this paper was presented in part in IEEE ICC 2012, Ottawa, Canada \cite{He_ICC12}.}
\thanks{$^\star$Fei He, Limin Xiao, Xiang Chen, and Shidong Zhou are with Department of Electronic Engineering, Research Institute of Information Technology, Tsinghua National Laboratory for Information Science and Technology (TNList), Tsinghua University, Beijing 100084, China. Xiang Chen, the corresponding author, who is also with Aerospace Center, Tsinghua University. E-mail: hef08@mails.tsinghua.edu.cn, \{xiaolm,chenxiang,zhousd\}@tsinghua.edu.cn.}
\thanks{$^\dag$Yin Sun is with the Department of Electrical and Computer Engineering, the Ohio State University, Columbus, Ohio 43210, USA. E-mail: sunyin02@gmail.com.}
\thanks{$^\ddagger$Chong-Yung Chi is with the Institute of Communications
Engineering and the Department of Electrical Engineering, National Tsinghua University, Hsinchu, Taiwan 30013. E-mail: cychi@ee.nthu.edu.tw.}}


\maketitle

\vspace{-1.5cm}
\begin{abstract}
In this paper, we consider two-way orthogonal frequency division multiplexing (OFDM) relay channels, where the direct link between the two terminal nodes is too weak to be used for data transmission. The widely known per-subcarrier decode-and-forward (DF) relay strategy, treats each subcarrier as a separate channel, and performs independent channel coding over
each subcarrier.
We show that this per-subcarrier DF relay strategy is only a suboptimal DF relay strategy, and present a multi-subcarrier DF relay strategy which utilizes cross-subcarrier channel coding to achieve a larger rate region.
We then propose an optimal resource allocation algorithm to characterize the achievable rate region of the multi-subcarrier DF relay strategy. The computational complexity of this algorithm is much smaller than that of standard Lagrangian duality optimization algorithms.
We further analyze the asymptotic performance of two-way relay strategies including the above two DF relay strategies and an amplify-and-forward (AF) relay strategy. The analysis shows that the multi-subcarrier DF relay strategy tends to achieve the capacity region of the two-way OFDM relay channels in the low signal-to-noise ratio (SNR) regime, while the AF relay strategy tends to achieve the multiplexing gain region of the two-way OFDM relay channels in the high SNR regime. Numerical results are provided to justify all the analytical results and the efficacy of the proposed optimal resource allocation algorithm.
\end{abstract}


\linespread{1.5}\selectfont 
\section{Introduction}\label{sec:intro}
Orthogonal frequency division multiplexing (OFDM) relaying is a cost-efficient technique to enhance the coverage and throughput of future wireless networks, and it has been widely advocated in many 4G standards, such as IEEE 802.16m and 3GPP advanced long term evolution (LTE-Advanced) \cite{Salem_SURV10,Peters_TCOMMAG09}. In practice, a relay node operates in a half-duplex mode to avoid strong self-interference. However, since the half-duplex relay node can not transmit all the time (or over the entire frequency band), the benefits provided by the relay node are not fully exploited \cite{Rankov_JSAC07}.

Recently, two-way relay technique has drawn extensive attention, because of its potential to improve the spectrum efficiency of one-way relay strategies \cite{Rankov_JSAC07,Larsson_VTC06S,Xie07,Kim_TIT08,Kim_TIT11,Tian_TWCOM12,
AsharK_Allerton12,Gunduz_Asilomar08,Vaze_TIT11,Sanguinetti_JSAC12}. If one utilizes traditional one-way relay strategies to realize two-way communications, four phases are needed. To improve the four-phase strategy, the two relay-to-destination phases can be combined into one broadcast phase \cite{Larsson_VTC06S,Xie07}, and the yielded three-phase strategy can support the same data rates with less channel resource by exploiting the side information at the terminal nodes. One can further combine the two source-to-relay phases into one multiple-access phase to yield a two-phase strategy (see Fig.~\ref{fig:two-way-OFDM}) \cite{Kim_TIT08}. Hybrid strategies with more phases have been considered in \cite{Kim_TIT11,Tian_TWCOM12,AsharK_Allerton12} to further enlarge the achievable rate region. The diversity-multiplexing tradeoff for two-way relay channels was studied in \cite{Gunduz_Asilomar08,Vaze_TIT11,Sanguinetti_JSAC12}.

Two-way relay strategies also have been in conjunction with OFDM techniques \cite{Ho_ICC08,Jang_TSP10,Gao_TSP09,Jitvan_TVT09,Liu_TWCOM10,Shin_VTC11F,Xiong_ICC12}. With amplify-and-forward (AF) relay strategy, power allocation and subcarrier permutation have been studied in \cite{Ho_ICC08,Jang_TSP10}, and its corresponding channel estimation problem has been thoroughly discussed in \cite{Gao_TSP09}. Resource allocation for two-way communications in an OFDM cellular network with both AF and decode-and-forward (DF) relay strategies was studies in \cite{Jitvan_TVT09}. A graph-based approach was proposed to solve the combinatorial resource allocation problem in \cite{Liu_TWCOM10}. For practical quality of service (QoS) requirements, the proportional fairness and transmission delay have been considered for two-way DF OFDM relay networks in \cite{Shin_VTC11F} and \cite{Xiong_ICC12}, respectively. All these studies of two-way OFDM relay channels with a DF strategy were almost centered on a per-subcarrier DF relay strategy, which treats each subcarrier as a separate two-way relay channel, and performs independent channel coding over each subcarrier. Such a per-subcarrier DF relay strategy is probably motivated by the fact that per-subcarrier channel coding can achieve the capacity of point-to-point OFDM channels. However, the story is different in OFDM relay channels: per-subcarrier channel coding can no longer attain the optimal achievable rate region of DF relaying for two-way OFDM relay channels. In other words, per-subcarrier DF relaying is merely a suboptimal DF relay strategy. More details are provided in Section~III, where an example is provided to show that a novel DF relay strategy achieves a larger rate region.
\begin{figure*}[t]
    \centering
    \subfigure[The multiple-access phase.]{
    \scalebox{0.85}{\includegraphics*[64,676][318,787]{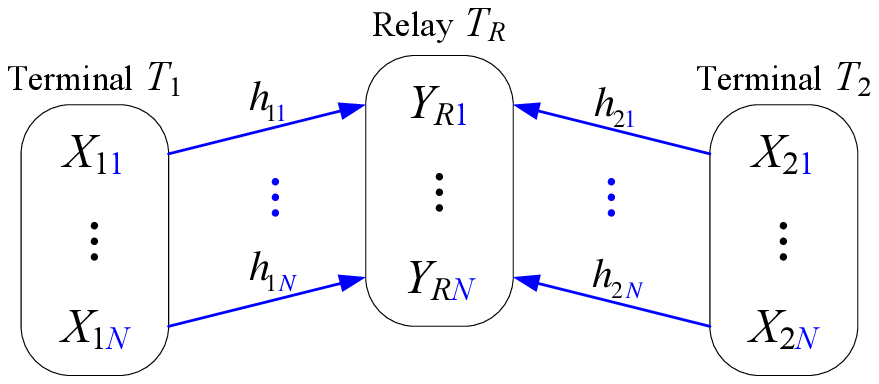}} \label{fig:MA-phase}}
    \hspace{0.4cm}
    \subfigure[The broadcast phase.]{
    \scalebox{0.85}{\includegraphics*[64,556][318,667]{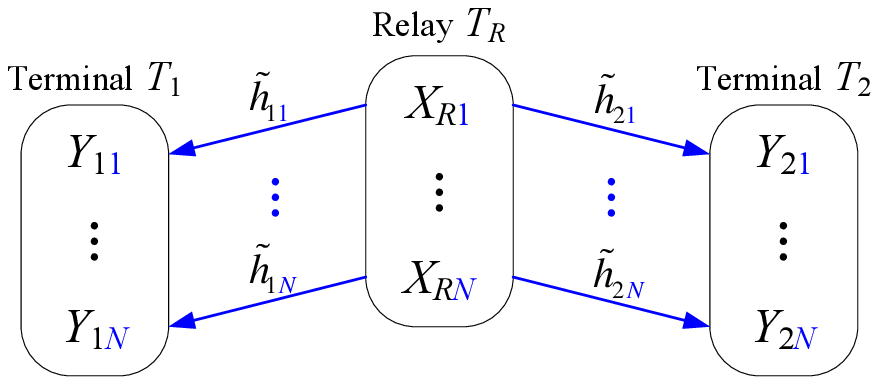}} \label{fig:BC-phase}}
    \caption{System model of a two-way OFDM relay channel, consisting of (a) a multiple-access phase and (b) a broadcast phase. } \label{fig:two-way-OFDM}
\end{figure*}

This paper focuses on the two-way OFDM relay channel with a negligible direct link due to large path attenuation or heavy blockage. This is motivated by the fact that the relay node plays a more important role when the direct link is weak than when it is strong \cite{Sanguinetti_JSAC12}. The optimal two-way relay strategy in this case consists of two phases, which are illustrated in Fig.~\ref{fig:two-way-OFDM}. We intend to answer the following questions in this paper: What is the optimal DF relay strategy when the direct link is negligible? Under what conditions is the optimal DF relay strategy better (or worse) than the AF relay strategy, and vice versa? Is the optimal DF relay strategy able to achieve the capacity region of two-way OFDM relay channels in some scenarios? To address these questions, we also investigate the capacity region bounds of the two-way OFDM relay channels and compare different relay strategies under the optimal resource allocation. The main contributions of this paper are summarized as follows:
\begin{itemize}
\item We present a multi-subcarrier DF relay strategy, which has a larger achievable rate region than the widely studied per-subcarrier DF relay strategy. Though this multi-subcarrier DF relay strategy is merely a simple extension of the existing result \cite{Kim_TIT08}, it is the optimal DF relay strategy for two-way OFDM relay channels.\footnote{A strategy is the optimal DF relay strategy, meaning that its achievable rate region contains the rate region of any other DF relay strategy. It is worth mentioning that relay strategies other than DF relay strategies may have a larger or smaller achievable rate region compared to this multi-subcarrier DF relay strategy in certain scenarios.} To the best of our knowledge, this multi-subcarrier DF relay strategy has not been reported in the open literature. We develop an optimal resource allocation algorithm to characterize the achievable rate region of the multi-subcarrier DF relay strategy. We show that the optimal resource allocation solution has a low-dimension structure. By exploiting this structure, the complexity of both primal and dual optimizations can be significantly reduced. The relative benefits of our multi-subcarrier DF relay strategy and its resource allocation algorithm are summarized in Table~\ref{tab:psc-msc}.
\item We analyze the asymptotic performance of different relay strategies in the low and high signal-to-noise ratio (SNR) regimes under optimal resource allocation. First, we show that the multi-subcarrier DF relay strategy tends to achieve the capacity region of two-way OFDM relay channels in the low SNR regime. Then, we characterize the multiplexing gain regions of the two DF relay strategies, the AF relay strategy, and the cut-set outer bound under optimal resource allocation. We show that the AF relay strategy can achieve the multiplexing gain region of two-way OFDM relay channels in the high SNR regime. Numerical results are provided to justify our analytical results. The asymptotic performance comparison of AF and DF strategies is summarized in Table~\ref{tab:AF-DF}.

\end{itemize}
\begin{table}[t]
\tabcolsep 0pt
\caption{Comparison of per-subcarrier and multi-subcarrier DF two-way relay strategies.}\label{tab:psc-msc}
\vspace*{-35pt}
\begin{center}
\def\temptablewidth{0.63\textwidth}
{\rule{\temptablewidth}{1pt}}
\begin{tabular*}{\temptablewidth}{@{\extracolsep{\fill}}c||c|c}
Strategy & Achievable rate region~~~ & Resource allocation complexity \\ \hline
per-subcarrier DF  & small & low \cite{Jitvan_TVT09}  \\
multi-subcarrier DF~ & large & very low
\end{tabular*}
{\rule{\temptablewidth}{1pt}}
\vspace*{-25pt}
\end{center}
\end{table}

The rest of this paper is organized as follows. Section~\ref{sec:system} presents the system model. Section~\ref{sec:region-msc} presents the multi-subcarrier two-way DF relay strategy and its achievable rate region. The resource allocation algorithm of the multi-subcarrier DF relay strategy is developed in Section~\ref{sec:RA-sol}. The asymptotic performance analysis of different relay strategies is provided in Section~\ref{sec:bound-analysis}. Some numerical results are presented in Section~\ref{sec:simulation}. Finally, Section~\ref{sec:conclusion} draws some conclusions.

{\bf Notation:} Throughout this paper, we use bold lowercase letters to denote column vectors, and also denote an $n\times1$ column vector by $(x_1,\ldots,x_n)$. $\Rset_+$ and $\Rset_+^n$ denote the set of nonnegative real numbers and the set of $n\times1$ column vectors with nonnegative real components, respectively. $\bm{p}\succeq\bm{0}$ means that each component of column vector $\bm{p}$ is nonnegative. $I(X;Y)$ denotes the mutual information between random variables $X$ and $Y$, and $I(X;Y|Z)$ denotes the conditional mutual information of $X$ and $Y$ given $Z$. $\mathbb{E}[\cdot]$ denotes the statistical expectation of the argument.

\begin{table}[t]
\tabcolsep 0pt
\caption{Asymptotic performance comparison of AF and DF two-way relay strategies.}\label{tab:AF-DF}
\vspace*{-35pt}
\begin{center}
\def\temptablewidth{0.68\textwidth}
{\rule{\temptablewidth}{1pt}}
\begin{tabular*}{\temptablewidth}{@{\extracolsep{\fill}}c||c|c}
Strategy & Low SNR & High SNR \\ \hline
multi-subcarrier DF~  & achieving capacity region~ & smaller multiplexing gain  \\
AF & lower rate &~ achieving multiplexing gain region
\end{tabular*}
{\rule{\temptablewidth}{1pt}}
\vspace*{-25pt}
\end{center}
\end{table}

\section{System Model} \label{sec:system}
We consider a two-way OFDM relay channel with $N$ subcarriers, where two terminal nodes $T_1$ and $T_2$ exchange messages by virtue of an intermediate relay node $T_R$.  The wireless transmissions in the two-way DF relay channel is composed of two phases: a multiple-access phase and a broadcast phase, as illustrated in Fig.~\ref{fig:two-way-OFDM}. In the multiple-access phase, the terminal nodes $T_1$ and $T_2$ simultaneously transmit their messages to the relay node $T_R$. In the broadcast phase, the relay node $T_R$ decodes its received messages, re-encodes them into a new codeword, and broadcasts it to the terminal nodes $T_1$ and $T_2$. The time proportion of the multiple-access phase is denoted as $t$ for $0<t<1$, and thereby the time proportion of the broadcast phase is $1-t$.

In the multiple-access phase, the received signal $Y_{Rn}$ of the relay node $T_R$ over subcarrier $n$ can be expressed as
\vspace{-4pt}
\begin{equation}
Y_{Rn}=h_{1n}\sqrt{\frac{p_{1n}}{t}}\,X_{1n}+h_{2n}\sqrt{\frac{p_{2n}}{t}}\,X_{2n}+Z_{Rn}, \label{eq:MA-channel}
\end{equation}
where $X_{in}$ ($i\In\{1,2\}$) is the unit-power transmitted symbol of the terminal nodes $T_i$ over subcarrier $n$, $h_{in}$ is the channel coefficient from $T_i$ to $T_R$ over subcarrier $n$, $p_{in}$ is the average transmission power, and $Z_{Rn}$ is the independent complex Gaussian noise with zero mean and variance $\sigma^2_{Rn}$.

In the broadcast phase, the received signals of the terminal nodes $T_1$ and $T_2$ over subcarrier $n$ are given by
\vspace{-6pt}
\begin{eqnarray}
&&Y_{1n}=\tilde{h}_{1n}\sqrt{\frac{p_{Rn}}{1-t}}\,X_{Rn}+Z_{1n}, \label{eq:BC-channel1}\\[-4pt]
&&Y_{2n}=\tilde{h}_{2n}\sqrt{\frac{p_{Rn}}{1-t}}\,X_{Rn}+Z_{2n}, \label{eq:BC-channel}
\end{eqnarray}
where $X_{Rn}$ and $p_{Rn}$ denote the unit-power transmitted symbol and the average transmission power of the relay node $T_R$ over subcarrier $n$, respectively, $\tilde{h}_{in}$ denotes the associated channel coefficient from $T_R$ to $T_i$ over subcarrier $n$, and $Z_{in}$ is the independent complex Gaussian noise with zero mean and variance $\sigma^2_{in}$ ($i\In\{1,2\}$).

Each node is subject to an individual average power constraint, which is given by
\vspace{-2pt}
\begin{equation} \label{eq:power-constr}
\sum_{n=1}^N{p_{in}}\leq P_i,~i=1,2,R,
\end{equation}
where $P_i$ denotes the maximum average transmission power of node $T_i$.
Let us use $\bm P{\triangleq}(P_1,P_2,P_R)$ to represent the maximum average powers of the three nodes, and use $\mathcal{G}\triangleq\{g_{1n},g_{2n},\tilde{g}_{1n},$ $\tilde{g}_{2n}\}_{n=1}^{N}$ to represent the channel state information (CSI), where $g_{in}{\triangleq}\,|h_{in}|^2/\sigma^2_{Rn}$ and $\tilde{g}_{in}{\triangleq}\,|\tilde{h}_{in}|^2/\sigma^2_{in}$ ($i\In\{1,2\}$) represent the normalized channel power gains. We assume that the perfect CSI $\mathcal{G}$ is available at the network controller to perform resource allocation throughout the paper.

\section{Optimal Two-Way OFDM DF Relay Strategy} \label{sec:region-msc}
\label{sec:rate-region}
This section presents a multi-subcarrier DF relay strategy, which can realize the optimal achievable rate region of the DF relay strategy for two-way OFDM relay channels. We also show that the per-subcarrier DF relay strategy considered in \cite{Jitvan_TVT09,Liu_TWCOM10,Shin_VTC11F,Xiong_ICC12} can only achieve a suboptimal rate region.

Let $R_{12}$ and $R_{21}$ denote the end-to-end data rates from $T_1$ to $T_2$ and from $T_2$ to $T_1$, respectively. When the direct link between $T_1$ and $T_2$ is negligible, the \emph{optimal} DF relay strategy of discrete memoryless two-way relay channels was given by Theorem~2 in \cite{Kim_TIT08}.
By applying this theorem to two-way parallel Gaussian relay channel and considering the optimal channel input distribution, we can obtain the optimal achievable rate region as stated in the following lemma:
\begin{Lemma} \label{prop1}
Given the maximum transmission powers $\bm P$ of the three nodes and the CSI $\mathcal{G}$, the optimal achievable rate region of the two-way parallel Gaussian relay channel \eqref{eq:MA-channel}-\eqref{eq:BC-channel} with a DF strategy is given by:
\vspace{-4pt}
{\setlength\arraycolsep{0pt} \begin{eqnarray}\label{eq:prop1}
\mathcal{R}_\text{DF}(\bm P,\mathcal{G})=\Bigg\{
&&(R_{12},R_{21})\In\Rset_+^2~\bigg|\nonumber\\[-4pt]
&&R_{12}\leq\min\!\bigg\{\sum_{n=1}^N{t\log_2\!\Big(1\!+\!\frac{g_{1n}p_{1n}}{t}\Big)},
\sum_{n=1}^N\!{(1-t)\log_2\!\Big(1\!+\!\frac{\tilde{g}_{2n}p_{Rn}}{1-t}\Big)}\bigg\},\nonumber\\[-4pt]
&&R_{21}\leq\min\!\bigg\{\sum_{n=1}^N{t\log_2\!\Big(1\!+\!\frac{g_{2n}p_{2n}}{t}\Big)},
\sum_{n=1}^N\!{(1-t)\log_2\!\Big(1\!+\!\frac{\tilde{g}_{1n}p_{Rn}}{1-t}\Big)}\bigg\},\nonumber\\[-4pt]
&&R_{12}+R_{21}\leq\sum_{n=1}^N{t\log_2\!\Big(1\!+\!\frac{g_{1n}p_{1n}+g_{2n}p_{2n}}{t}\Big)},\nonumber\\[-4pt]
&&0<t<1,\,\sum_{n=1}^N{p_{in}}\leq P_i,\,p_{in}\geq0,\,i=1,2,R,\,n=1,\ldots,N\Bigg\}.
\end{eqnarray}}
\end{Lemma}
\begin{proof}
See Appendix~\ref{sec:proof-prop1}.
\end{proof}

In fact, the optimal rate region of \eqref{eq:prop1} is the intersection of the capacity regions of a parallel multi-access channel and a parallel broadcast channel with receiver side information\footnote{Here, the receiver side information means
each user's own transmitted message.} \cite{Kramer_ITW07}. This rate region can be achieved by the following multi-subcarrier DF relay strategy: In the multiple-access phase, the relay node decodes the messages from the two terminal nodes by either successive cancellation decoding with time sharing/rate-splitting, or joint decoding \cite{Cover06,Gamal11,Tse_TIT98}. In the broadcast phase, the relay node can utilize nested lattice codes, nested and algebraic superposition codes to transmit the messages to the intended destinations \cite{Kramer_ITW07,Tian_TWCOM12}. Some related information theoretical random coding techniques were discussed in \cite{Xie07,Kramer_ITW07,Oechtering_TIT08}. In either of the phases, channel encoding/decoding is performed jointly across all the subcarriers.

On the other hand, the per-subcarrier DF relay strategy independently implements the DF relay scheme of \cite{Kim_TIT08} over each subcarrier \cite{Jitvan_TVT09,Liu_TWCOM10,Shin_VTC11F,Xiong_ICC12}. The achievable rate region of the per-subcarrier two-way DF relay strategy is given by
\vspace{-4pt}
{\setlength\arraycolsep{0pt} \begin{eqnarray}\label{eq:region-psc}
\mathcal{R}_\text{p,DF}(\bm P,\mathcal{G})=\Bigg\{
&&(R_{12},R_{21})\In\Rset_+^2~\bigg|\nonumber\\[-4pt]
&&R_{12}\leq\sum_{n=1}^N{\min\!\bigg\{t\log_2\!\Big(1\!+\!\frac{g_{1n}p_{1n}}{t}\Big),
(1-{t})\log_2\!\Big(1\!+\!\frac{\tilde{g}_{2n}p_{Rn}}{1-t}\Big)\bigg\}},\nonumber\\[-4pt]
&&R_{21}\leq\sum_{n=1}^N{\min\!\bigg\{t\log_2\!\Big(1\!+\!\frac{g_{2n}p_{2n}}{t}\Big),
(1-{t})\log_2\!\Big(1\!+\!\frac{\tilde{g}_{1n}p_{Rn}}{1-t}\Big)\bigg\}},\nonumber\\[-4pt]
&&R_{12}+R_{21}\leq\sum_{n=1}^N{t\log_2\!\Big(1\!+\!\frac{g_{1n}p_{1n}+g_{2n}p_{2n}}{t}\Big)},\nonumber\\[-4pt]
&&0<t<1,\,\sum_{n=1}^N{p_{in}}\leq P_i,\,p_{in}\geq0,\,i=1,2,R,\,n=1,\ldots,N\Bigg\}.
\end{eqnarray}}
$\!\!$The only difference between ${\mathcal{R}}_\text{DF}(\bm P,\mathcal{G})$ and ${\mathcal{R}}_\text{p,DF}(\bm P,\mathcal{G})$ lies in the order of the function $\min\{\cdot\}$ and the summation in \eqref{eq:prop1} and \eqref{eq:region-psc}, implying ${\mathcal{R}}_\text{p,DF}(\bm P,\mathcal{G})\subseteq{\mathcal{R}}_\text{DF}(\bm P,\mathcal{G})$.
Therefore, the per-subcarrier DF relay strategy is only a suboptimal DF relay strategy. Similar results have been reported in \cite{Liang_TIT07,Sun_TSP12} for one-way parallel relay channels.
\ifreport
\begin{figure}[t]
    \centering
    \scalebox{0.6}{\includegraphics*[83,477][405,795]{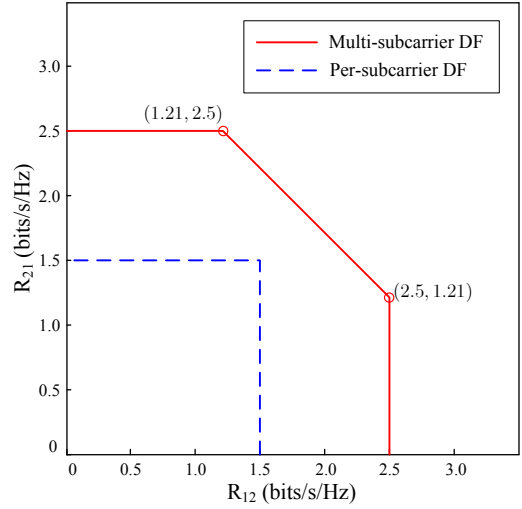}}
    \caption{Comparison of the achievable rate regions of different two-way OFDM relay strategies for a toy example with $N=2,\,(g_{11},g_{12},g_{21},g_{22})=(1,15,7,3),\,\tilde{g}_{in}=g_{in},\,p_{in}=0.5,\,t=0.5$.} \label{fig:region-compare-example}
\end{figure}
\else
\fi

We now provide a toy example to compare these two DF relay strategies. Consider a two-way OFDM relay channel with $N=2$ subcarriers. The wireless channel power gains are given by $(g_{11},g_{12},g_{21},g_{22})$ $=(1,15,7,3)$ and $\tilde{g}_{in}=g_{in}$ for $n,i\in\{1,2\}$. The power and channel resources are fixed to be $p_{in}=0.5$ and $t=0.5$. According to Lemma~\ref{prop1}, the achievable rate region of the multi-subcarrier DF relay strategy with fixed resource allocation is given by the set of rate pairs $(R_{12},R_{21})$ satisfying
\vspace{-6pt}
\begin{subequations}\label{eq:region1}
\begin{align}
&R_{12}\leq\min\{0.5+2,1.5+1\}=2.5~\text{bits/s/Hz},\\[-6pt]
&R_{21}\leq\min\{1.5+1,0.5+2\}=2.5~\text{bits/s/Hz},\\[-6pt]
&R_{12}+R_{21}\leq0.5\left[\log_2(9)+\log_2(19)\right]=3.71~\text{bits/s/Hz}.\\[-30pt]\nonumber
\end{align}
\end{subequations}
Similarly, by \eqref{eq:region-psc}, the achievable rate region of the per-subcarrier DF relay strategy with fixed resource allocation is given by the set of rate pairs $(R_{12},R_{21})$ satisfying
\vspace{-6pt}
\begin{subequations}\label{eq:region2}
\begin{align}
&R_{12}\leq\min\{0.5,1.5\}+\min\{2,1\}=1.5~\text{bits/s/Hz},\\[-6pt]
&R_{21}\leq\min\{1.5,0.5\}+\min\{1,2\}=1.5~\text{bits/s/Hz},\\[-6pt]
&R_{12}+R_{21}\leq0.5\left[\log_2(9)+\log_2(19)\right]=3.71~\text{bits/s/Hz},\\[-30pt]\nonumber
\end{align}
\end{subequations}
where the sum-rate constraint is actually inactive.
\ifreport
The achievable rate regions in \eqref{eq:region1} and \eqref{eq:region2} are shown in Fig.~\ref{fig:region-compare-example}, from which one can easily observe that the considered multi-subcarrier DF relay strategy has a larger achievable rate region.
\fi
$\!\!$An effective and computationally efficient approach for the optimal resource allocation of the proposed two-way DF strategy will be presented in the next section.

\section{Resource Allocation Algorithm} \label{sec:RA-sol}
We now develop a resource allocation algorithm to characterize the boundary of the achievable rate region ${\mathcal{R}}_\text{DF}(\bm P,\mathcal{G})$ in \eqref{eq:prop1}. We will show that the optimal resource allocation solution has a low-dimension structure, and thereby the number of dual variables to be optimized is reduced; see Propositions~\ref{prop2} and \ref{prop3} below for more details. The complexity of our resource allocation algorithm turns out to be much lower than that of the standard Lagrangian dual optimization algorithm and the existing resource allocation algorithm reported in \cite{Jitvan_TVT09}.

\subsection{Resource Allocation Problem Formulation}
Let $\rho\in(0,\infty)$ denote the rate ratio of the two terminal nodes, i.e.,
\vspace{-6pt}
\begin{equation} \label{eq:rate-ratio}
\rho\triangleq R_{21}/R_{12}.
\end{equation}
Then, a boundary point $(R_{12},R_{21})=(R_{12},{\rho}R_{12})$ of the achievable rate region ${\mathcal{R}}_\text{DF}(\bm P,\mathcal{G})$ is attained by maximizing $R_{12}$ within ${\mathcal{R}}_\text{DF}(\bm P,\mathcal{G})$ for a given rate ratio $\rho$. Therefore, the boundary point of ${\mathcal{R}}_\text{DF}(\bm P,\mathcal{G})$ is characterized by the following resource allocation problem:
\vspace{-7pt}
\begin{subequations} \label{eq:original-problem}
\begin{align}
\hspace{-0.4cm}\max_{\bm{p}_1,\bm{p}_2,\bm{p}_R\succeq\bm{0},\,R_{12},\,t}\quad&R_{12}\\[-6pt]
\text{s.t.}~~~~~~\quad&R_{12}\leq t\sum_{n=1}^N{\log_2\!\Big(1\!+\!\frac{g_{1n}p_{1n}}{t}\Big)},\label{eq:mac-ieq1}\\[-4pt]
&R_{12}\leq\frac{t}{\rho}\sum_{n=1}^N{\log_2\!\Big(1\!+\!\frac{g_{2n}p_{2n}}{t}\Big)},\label{eq:mac-ieq2}\\[-4pt]
&R_{12}\leq\frac{t}{\rho+1}\sum_{n=1}^N{\log_2\!\Big(1\!+\!\frac{g_{1n}p_{1n}+g_{2n}p_{2n}}{t}\Big)},\label{eq:mac-ieq3}\\[-4pt]
&R_{12}\leq(1-{t})\sum_{n=1}^N{\log_2\!\Big(1\!+\!\frac{\tilde{g}_{2n}p_{Rn}}{1-t}\Big)},\label{eq:bc-ieq1}\\[-4pt]
&R_{12}\leq\frac{1-{t}}{\rho}\sum_{n=1}^N{\log_2\!\Big(1\!+\!\frac{\tilde{g}_{1n}p_{Rn}}{1-t}\Big)},\label{eq:bc-ieq2}\\[-6pt]
&\sum_{n=1}^N{p_{in}}\leq P_i,\;i=1,2,R,\\[-6pt]
&0<t<1,
\end{align}
\end{subequations}
where $\bm{p}_i\triangleq(p_{i1},p_{i2},\ldots,p_{iN})\In\Rset_{+}^N$ denotes the power allocation of node $T_i$ for $i{=}1,2,R$. Problem \eqref{eq:original-problem} is a convex optimization problem, which can be solved by standard interior-point methods or by using general purpose convex solvers such as CVX \cite{Boyd09_CVX}. However, these methods quickly become computationally formidable as the number of subcarriers $N$ increases, because their complexity grows in the order of $O(N^{3.5})$ \cite{Karmarkar84}, \cite[p.~8 and Eq.~(11.29)]{Boyd04}. Since $N$ can be quite large in practical OFDM systems, we will develop a more efficient resource allocation algorithm for large values of $N$ in the sequel.

\subsection{Phase-Wise Decomposition of Problem~\eqref{eq:original-problem}}
\label{sec:phase-decomp}
Let us first fix the value of $t$. Then, problem \eqref{eq:original-problem} can be decomposed into two power allocation subproblems for the multi-access phase and the broadcast phase, respectively. Note that the transmission powers of the terminal nodes $\bm{p}_1$ and $\bm{p}_2$ are only involved in the rate constraints \eqref{eq:mac-ieq1}-\eqref{eq:mac-ieq3} for the multiple-access phase, while the transmission power of the relay node $\bm{p}_R$ is only involved in the rate constraints \eqref{eq:bc-ieq1} and \eqref{eq:bc-ieq2} for the broadcast phase. Let $R_\text{MA}$ and $R_\text{BC}$ denote the achievable rates for the multiple-access and broadcast phases, respectively. For any fixed $t$, problem \eqref{eq:original-problem} can be decomposed into the following two subproblems, one for the multiple-access phase
\vspace{-4pt}
\begin{subequations} \label{eq:MA-problem}
\begin{align}
R_\text{MA}^{\star}(t)\triangleq\max_{\bm{p}_1,\bm{p}_2\succeq\bm{0},\,R_\text{MA}}\quad&R_\text{MA} \\[-6pt]
\st~~~~\quad&R_\text{MA}\leq r_k(\bm{p}_1,\bm{p}_2),\;k=1,2,3,\label{eq:rate-constr-MA}\\[-6pt]
&\sum_{n=1}^N{p_{in}}\leq P_i,\;i=1,2, \label{eq:power-constr-MA}
\end{align}
\end{subequations}
and the other for the broadcast phase
\vspace{-4pt}
\begin{subequations} \label{eq:BC-problem}
\begin{align}
R_\text{BC}^{\star}(t)\triangleq\max_{\bm{p}_R\succeq\bm{0},\,R_\text{BC}}\quad&R_\text{BC}\\[-6pt]
\st~~~\quad&R_\text{BC}\leq r_k(\bm{p}_R),\;k=4,5,\label{eq:rate-constr-BC}\\[-6pt]
&\sum_{n=1}^N{p_{Rn}}\leq P_R, \label{eq:power-constr-BC}
\end{align}
\end{subequations}
where the rate functions $r_k(\bm{p}_1,\bm{p}_2),\,k=1,2,3$, and $r_k(\bm{p}_R),\,k=4,5$, are defined by
\vspace{-3pt}
\begin{subequations}
\begin{align}
&r_1(\bm{p}_1,\bm{p}_2)= t\sum_{n=1}^N{\log_2\!\Big(1\!+\!\frac{g_{1n}p_{1n}}{t}\Big)},\label{eq:rate-fun1}\\[-4pt]
&r_2(\bm{p}_1,\bm{p}_2)= \frac{t}{\rho}\sum_{n=1}^N{\log_2\!\Big(1\!+\!\frac{g_{2n}p_{2n}}{t}\Big)},\label{eq:rate-fun2}\\[-4pt]
&r_3(\bm{p}_1,\bm{p}_2)= \frac{t}{\rho+1}\sum_{n=1}^N{\log_2\!\Big(1\!+\!\frac{g_{1n}p_{1n}+g_{2n}p_{2n}}{t}\Big)},\label{eq:rate-fun3}\\[-4pt]
&r_4(\bm{p}_R)= (1-{t})\sum_{n=1}^N{\log_2\!\Big(1\!+\!\frac{\tilde{g}_{2n}p_{Rn}}{1-t}\Big)},\label{eq:rate-fun4}\\[-4pt]
&r_5(\bm{p}_R)= \frac{1-{t}}{\rho}\sum_{n=1}^N{\log_2\!\Big(1\!+\!\frac{\tilde{g}_{1n}p_{Rn}}{1-t}\Big)}.\label{eq:rate-fun5}
\end{align}
\end{subequations}

Then, the optimal objective value of problem \eqref{eq:original-problem} is given by
\vspace{-4pt}
\begin{equation} \label{eq:phase-decomp-mu}
R_{12}^{\star}=\max_{0<t<1}\min\!\left\{R_\text{MA}^{\star}(t),R_\text{BC}^{\star}(t)\right\},
\end{equation}
where $R_\text{MA}^{\star}(t)$ and $R_\text{BC}^{\star}(t)$ are defined in \eqref{eq:MA-problem} and \eqref{eq:BC-problem}, respectively. Since problem \eqref{eq:phase-decomp-mu} itself is a one-dimensional convex optimization problem, it can be efficiently solved by either golden section search method or the bisection method \cite[Chapter~8]{Bazaraa06}, with $R_\text{MA}^{\star}(t)$ and $R_\text{BC}^{\star}(t)$ at each search iteration obtained by solving \eqref{eq:MA-problem} and \eqref{eq:BC-problem}, respectively. Next, let us show how to solve the subproblems \eqref{eq:MA-problem} and \eqref{eq:BC-problem}, respectively.

\subsection{Lagrange Dual Optimization for Subproblem \eqref{eq:MA-problem}} \label{sec:dual-decomp}
Let us define the partial Lagrange dual function of subproblem \eqref{eq:MA-problem} as
\vspace{-4pt}
\begin{equation} \label{eq:dual-function}
D_\text{MA}(\bm{\lambda},\bm{\alpha})\triangleq \min_{\bm{p}_1,\bm{p}_2\succeq\bm{0},\,R_\text{MA}}L_\text{MA}\left(\bm{p}_1,\bm{p}_2,R_\text{MA},\bm{\lambda},\bm{\alpha}\right),
\end{equation}
where $\bm{\lambda}=(\lambda_1,\lambda_2,\lambda_3),\,\bm{\alpha}=(\alpha_1,\alpha_2)$ are nonnegative dual variables associated with three rate inequality constraints in \eqref{eq:rate-constr-MA} and two power inequality constraints in \eqref{eq:power-constr-MA}, respectively, and\!
\vspace{-4pt}
{\setlength\arraycolsep{1pt} \begin{eqnarray} \label{eq:Lagrangian0}
\hspace{-9pt}L_\text{MA}\left(\bm{p}_1,\bm{p}_2,R_\text{MA},\bm{\lambda},\bm{\alpha}\right)&&=-R_\text{MA}
+\sum_{k=1}^3{\lambda_k\big[R_\text{MA}-r_k(\bm{p}_1,\bm{p}_2)\big]}
+\sum_{i=1}^2{\alpha_i\Big(\sum_{n=1}^N{p_{in}}-P_i\Big)}
\end{eqnarray}}
$\!\!$is the partial Lagrangian of \eqref{eq:MA-problem}. Then, the corresponding dual problem is defined as
\vspace{-9pt}
\begin{equation} \label{eq:dual-problem}
\max_{\bm{\lambda}\succeq\bm{0},\,\bm{\alpha}\succeq\bm{0}}\,D_\text{MA}(\bm{\lambda},\bm{\alpha}).
\end{equation}
Since the refined Slater's condition \cite[Eq.~(5.27)]{Boyd04} is satisfied in problem \eqref{eq:MA-problem}, the duality gap between problems \eqref{eq:MA-problem} and \eqref{eq:dual-problem} is zero, i.e., solving problem \eqref{eq:dual-problem} in the dual domain will yield the optimal solution of the primal problem \eqref{eq:MA-problem}.

\subsubsection{Structure of the Optimal Dual Solution $\bm{\lambda}^{\star}$}
Prior to the presentation of our power allocation algorithm for solving the problems \eqref{eq:dual-function} and \eqref{eq:dual-problem}, we first present an important result that the optimal dual solution $\bm{\lambda}^{\star}$ satisfies the following structural property:
\begin{Proposition} \label{prop2}
There exists one optimal solution $(\bm{\lambda}^{\star},\bm{\alpha}^{\star})$ to the dual problem \eqref{eq:dual-problem}, where
$\bm{\lambda}^{\star}=(1{-}\lambda_3^{\star},0,\lambda_3^{\star})$ or $\bm{\lambda}^{\star}=(0,1{-}\lambda_3^{\star},\lambda_3^{\star})$ and $0\leq\lambda_3^{\star}\leq 1$.
\end{Proposition}
\noindent\emph{Proof:}
See Appendix~\ref{sec:proof-prop2}. \hfill{$\blacksquare$}

Proposition~\ref{prop2} is very useful for developing our power allocation algorithm, because the search region for $\bm{\lambda}^{\star}$ can be confined to a set $\bm{\Lambda}_1\bigcup\bm{\Lambda}_2$, where $\bm{\Lambda}_1$ and $\bm{\Lambda}_2$ are two one-dimensional dual sets defined by
\vspace{-9pt}
\begin{subequations} \label{eq:lambda-set}
\begin{align}
\bm{\Lambda}_1\triangleq\{\bm{\lambda}\In\Rset_+^3\mid\bm{\lambda}= (1-\lambda_3,0,\lambda_3),\;0\leq\lambda_3\leq1\},\label{eq:lambda-set1}\\[-6pt]
\bm{\Lambda}_2\triangleq\{\bm{\lambda}\In\Rset_+^3\mid\bm{\lambda}= (0,1-\lambda_3,\lambda_3),\;0\leq\lambda_3\leq1\}.\label{eq:lambda-set2}\\[-30pt]\nonumber
\end{align}
\end{subequations}
In the sequel, we will show that finding solutions to both problems \eqref{eq:dual-function} and \eqref{eq:dual-problem} can be substantially simplified by virtue of Proposition~\ref{prop2}.

\subsubsection{Primal Solution to Problem \eqref{eq:dual-function}} \label{sec:Lagrangian-n}
As the first important application of Proposition~\ref{prop2}, we show that the structure of $\bm{\lambda}^{\star}$ can be exploited to simplify the primal solution to problem \eqref{eq:dual-function}. For any given dual variables $(\bm{\lambda},\bm{\alpha})$, the optimal power allocation solution $(p_{1n}^{\star},p_{2n}^{\star})$ to problem \eqref{eq:dual-function} is determined by the following Karush-Kuhn-Tucker (KKT) conditions:
\vspace{-4pt}
\begin{subequations} \label{eq:MA-KKT}
\begin{align}
\frac{\partial{L_{\text{MA}}}}{\partial{p_{1n}}}=
\alpha_1\!-\!\frac{t{g_{1n}}\lambda_3/(\rho+1)}{(t+g_{1n}p_{1n}^{\star}+g_{2n}p_{2n}^{\star})\ln2}
-\frac{t{g_{1n}}\lambda_1}{(t+g_{1n}p_{1n}^{\star})\ln2}
&\;\left\{\begin{array}{ll}
\geq0,&\;\text{if}\;p_{1n}^{\star}=0; \\[2pt]
=0,&\;\text{if}\;p_{1n}^{\star}>0,
\end{array}\right. \label{eq:MA-KKT1}\\
\frac{\partial{L_{\text{MA}}}}{\partial{p_{2n}}}=
\alpha_2\!-\!\frac{t{g_{2n}}\lambda_3/(\rho+1)}{(t+g_{1n}p_{1n}^{\star}+g_{2n}p_{2n}^{\star})\ln2}
-\!\frac{t{g_{2n}}\lambda_2/\rho}{(t+g_{2n}p_{2n}^{\star})\ln2}
&\;\left\{\begin{array}{ll}
\geq0,&\;\text{if}\;p_{2n}^{\star}=0; \\[2pt]
=0,&\;\text{if}\;p_{2n}^{\star}>0.
\end{array}\right. \label{eq:MA-KKT2}
\end{align}
\end{subequations}

According to Proposition~\ref{prop2}, at least one of $\lambda_1^{\star}$ and $\lambda_2^{\star}$ is 0, which can be utilized to simplify the KKT conditions \eqref{eq:MA-KKT}. The attained optimal $(p_{1n}^{\star},p_{2n}^{\star})$ is provided in the following four cases:

\textbf{Case~1}: $p_{1n}^{\star}>0,\,p_{2n}^{\star}>0$.
If $\bm{\lambda}=(1{-}\lambda_3,0,\lambda_3)$, then
\vspace{-6pt}
\begin{subequations} \label{eq:pa-case1}
\begin{align}
p_{1n}^{\star}&=~\frac{t(1-\lambda_3)}{(\alpha_1-\frac{g_{1n}}{g_{2n}}\alpha_2)\ln2}-\!\frac{t}{g_{1n}},\\[-4pt]
p_{2n}^{\star}&=~\frac{t\lambda_3}{(\rho+1)\alpha_2\ln2}-\frac{t(1-\lambda_3)}{(\frac{g_{2n}}{g_{1n}}\alpha_1-\alpha_2)\ln2};
\end{align}
\end{subequations}
otherwise, if $\bm{\lambda}=(0,1{-}\lambda_3,\lambda_3)$, then
\vspace{-6pt}
\begin{subequations} \label{eq:pa-case1-2}
\begin{align}
p_{1n}^{\star}&=~\frac{t\lambda_3}{(\rho+1)\alpha_1\ln2}-\frac{t(1-\lambda_3)}{\rho(\frac{g_{1n}}{g_{2n}}\alpha_2-\alpha_1)\ln2},\\[-6pt]
p_{2n}^{\star}&=~\frac{t(1-\lambda_3)}{\rho(\alpha_2-\frac{g_{2n}}{g_{1n}}\alpha_1)\ln2}-\!\frac{t}{g_{2n}}.
\end{align}
\end{subequations}
Case 1 happens if $p_{1n}^{\star}p_{2n}^{\star}>0$ is satisfied in \eqref{eq:pa-case1} or \eqref{eq:pa-case1-2}.

\textbf{Case~2}: $p_{1n}^{\star}>0,\,p_{2n}^{\star}=0$. Then the solutions to \eqref{eq:MA-KKT1} and \eqref{eq:MA-KKT2} are given by
\vspace{-2pt}
\begin{subequations} \label{eq:pa-case2}
\begin{align}
p_{1n}^{\star}&=~\frac{t[(\rho+1)\lambda_1+\lambda_3]}{(\rho+1)\alpha_1\ln2}-\!\frac{t}{g_{1n}}, \label{eq:pa-case2-p1}\\[-9pt]
p_{2n}^{\star}&=~0.
\end{align}
\end{subequations}
This case happens if $p_{1n}^{\star}>0$ is satisfied in \eqref{eq:pa-case2-p1}.

\textbf{Case~3}: $p_{1n}^{\star}=0,\,p_{2n}^{\star}>0$. Then the solutions to \eqref{eq:MA-KKT1} and \eqref{eq:MA-KKT2} are given by
\vspace{-9pt}
\begin{subequations} \label{eq:pa-case3}
\begin{align}
p_{1n}^{\star}&=~0,\\[-4pt]
p_{2n}^{\star}&=~\frac{t[(\rho+1)\lambda_2+\rho\lambda_3]}{\rho(\rho+1)\alpha_2\ln2}-\!\frac{t}{g_{2n}}. \label{eq:pa-case3-p2}
\end{align}
\end{subequations}
This case happens if $p_{2n}^{\star}>0$ is satisfied in \eqref{eq:pa-case3-p2}.

\textbf{Case~4}: $p_{1n}^{\star}=0,\,p_{2n}^{\star}=0$. This is the default case when the above 3 cases do not happen.

\noindent\textbf{Remark~1}~
If the structural property of $\bm{\lambda}^{\star}$ in Proposition~\ref{prop2} is not available, one can still obtain an alternative closed-form solution to \eqref{eq:MA-KKT} \cite{He_ICC12}. However, this solution involves solving a more difficult cubic equation, which is presented in Appendix~\ref{sec:cubic-eq}. Nevertheless, these two closed-form solutions are much simpler than the iterative power allocation procedure proposed in \cite{Jitvan_TVT09} for the per-subcarrier DF relay strategy.

\ifreport
\noindent\textbf{Remark~2}~
Since the Lagrangian \eqref{eq:Lagrangian0} is not strictly convex with respect to the primal power variables at some dual points, the power allocation solution in \eqref{eq:pa-case1}-\eqref{eq:pa-case3} may be non-unique at those dual points. For example, if $\lambda_3=1$ and $\alpha_1=\frac{g_{1n}}{g_{2n}}\alpha_2$, the power allocation solution in either of \eqref{eq:pa-case1} and \eqref{eq:pa-case1-2} is indeterminate due to the presence of $\frac{\,0}{\,0}$ form. In fact, any nonnegative $(p_{1n}^{\star},p_{2n}^{\star})$ satisfying $g_{1n}p_{1n}^{\star}+g_{2n}p_{2n}^{\star}=\frac{t g_{1n}}{\alpha_1(\rho+1)\ln2}-1$ is a solution to \eqref{eq:MA-KKT} in this case. Nevertheless, any one of the optimal primal power solutions can be used to derive the subgradient for solving the dual problem \eqref{eq:dual-problem} \cite[Section~6.1]{Bertsekas99}. After the optimal dual point $(\bm{\lambda}^{\star},\bm{\alpha}^{\star})$ is obtained, we still need to recover a feasible solution to the primal problem \eqref{eq:MA-problem} \cite{Bertsekas99,Xiao_TCOM04,Yu_TCOM06}. According to \cite[Proposition~5.1.1]{Bertsekas99}, this can be accomplished by incorporating the feasibility conditions \eqref{eq:MA-KKT-feasible1} and the complementary slackness conditions \eqref{eq:MA-KKT-slackness1}, which are contained in the KKT conditions of problem \eqref{eq:MA-problem} given in Appendix~\ref{sec:proof-prop2}, into the power allocation solutions \eqref{eq:pa-case1}-\eqref{eq:pa-case3}, which involves solving a system of linear equations and inequalities.
\fi

\subsubsection{Dual Solution to Problem \eqref{eq:dual-problem}} \label{sec:dual-max}
We now solve the dual problem \eqref{eq:dual-problem} by a two-level optimization approach \cite{Jitvan_TVT09}, which first fixes $\bm{\lambda}$ and searches for the optimal solution $\bm{\alpha}^{\star}(\bm{\lambda})$ to the maximization problem
\vspace{-4pt}
\begin{equation} \label{eq:alpha-max}
G_\text{MA}(\bm{\lambda})\triangleq\max_{\bm{\alpha}\succeq\bm{0}}\,D_\text{MA}(\bm{\lambda},\bm{\alpha}),
\end{equation}
and then optimizes $\bm{\lambda}$ by
\vspace{-6pt}
\begin{equation} \label{eq:lambda-max}
\bm{\lambda}^{\star}\triangleq\arg\max_{\bm{\lambda}\succeq\bm{0}}\,G_\text{MA}(\bm{\lambda}).
\end{equation}
The inner-level optimization problem \eqref{eq:alpha-max} is solved by an ellipsoid method \cite{Boyd07_ellip} summarized in Algorithm~\ref{alg1}, where the subgradient of the dual problem $D_\text{MA}(\bm{\lambda},\bm{\alpha})$ with respect to $\bm{\alpha}$ is given by \cite[Proposition~6.1.1]{Bertsekas99}
\vspace{-4pt}
\begin{equation} \label{eq:subgradient-alpha}
\bm{\eta}(\bm{\lambda},\bm{\alpha})=\left(\sum_{n=1}^N{p_{1n}^{\star}}-P_1,\sum_{n=1}^N{p_{2n}^{\star}}-P_2\right),
\end{equation}
where $(p_{1n}^{\star},p_{2n}^{\star})$ is the optimal power allocation solution obtained by \eqref{eq:pa-case1}-\eqref{eq:pa-case3}.
\ifreport
More details about the initialization of $\bm{\alpha}$ and the matrix $\mathbf{A}$ in Algorithm~\ref{alg1} are given in Appendix~\ref{sec:alpha-bounds}.
\fi

By Proposition~\ref{prop2}, the outer-level optimization problem \eqref{eq:lambda-max} can be solved by searching for $\bm{\lambda}^{\star}$ over the set $\bm{\Lambda}_1\bigcup\bm{\Lambda}_2$, i.e.,
\vspace{-9pt}
\begin{equation} \label{eq:lambda-max1}
\bm{\lambda}^{\star}=\arg\max_{\bm{\lambda}\In\bm{\Lambda}_1\bigcup\bm{\Lambda}_2}\,G_\text{MA}(\bm{\lambda}).
\end{equation}

In order to solve the reduced outer-level optimization problem \eqref{eq:lambda-max1}, we first need the subgradient of the objective function $G_\text{MA}(\bm{\lambda})$.
According to \cite{Bertsekas99} and \cite[Corollary~4.5.3]{Hiriart01}, one subgradient of $G_\text{MA}(\bm{\lambda})$ in \eqref{eq:alpha-max} is given by
\vspace{-4pt}
\begin{equation} \label{eq:subgradient-lambda}
\bm{\xi}(\bm{\lambda})= \left(R_\text{MA}^{\star}-r_1,R_\text{MA}^{\star}-r_2,R_\text{MA}^{\star}-r_3\right),
\end{equation}
where $R_\text{MA}^{\star}=\min\{r_1,r_2,r_3\}$, and $r_k~(k{=}1,2,3)$ are the rate functions \eqref{eq:rate-fun1}-\eqref{eq:rate-fun3} associated with the optimal primal power allocation solution \eqref{eq:pa-case1}-\eqref{eq:pa-case3} obtained at the dual point $(\bm{\lambda},\bm{\alpha}^{\star}(\bm{\lambda}))$, respectively, and $\bm{\alpha}^{\star}(\bm{\lambda})$ is the optimal solution to \eqref{eq:alpha-max}.

With the subgradient $\bm{\xi}(\bm{\lambda})$ of $G_\text{MA}(\bm{\lambda})$, we are ready to solve the outer-level optimization problem \eqref{eq:lambda-max1}. Instead of searching both the sets $\bm{\Lambda}_1$ and $\bm{\Lambda}_2$, we propose a simple testing method to determine whether $\bm{\lambda}^{\star}\In\bm{\Lambda}_1$ or $\bm{\lambda}^{\star}\In\bm{\Lambda}_2$. Noticing that $\bm{\Lambda}_1\bigcap\bm{\Lambda}_2\,{=}\,\{(0,0,1)\}$, let us consider a testing method at the dual point $\bm{\lambda}^0=(0,0,1)$. By the concavity of the dual function $D_\text{MA}(\bm{\lambda},\bm{\alpha})$, $G_\text{MA}(\bm{\lambda})$ is also concave in $\bm{\lambda}$, which implies \cite[Eq.~(B.21)]{Bertsekas99}
\begin{algorithm}
\caption{The ellipsoid method for solving the inner-level problem \eqref{eq:alpha-max}} \label{alg1}
\begin{algorithmic}[1]
\State \textbf{Input} CSI $\{g_{1n},g_{2n}\}_{n=1}^{N}$, average powers $\{P_1,P_2\}$, rate ratio $\rho$, time proportion $t$, and $\bm{\lambda}$.
\State Initialize $\bm{\alpha}$ and a $2\times2$ positive definite matrix $\mathbf{A}$ that define the ellipsoid $\mathcal{E}(\bm{\alpha},\mathbf{A})=\{\bm{x}\In\Rset_+^2\mid(\bm{x}-\bm{\alpha})^T
\mathbf{A}^{-1}(\bm{x}-\bm{\alpha})\leq1\}$.
\Repeat
\State Compute the optimal $(p_{1n}^\star,p_{2n}^\star)$ by \eqref{eq:pa-case1}-\eqref{eq:pa-case3} for given $(\bm{\lambda},\bm{\alpha})$ and $t$.
\State Compute the subgradient $\bm{\eta}(\bm{\lambda},\bm{\alpha})$ with respect to $\bm\alpha$ by \eqref{eq:subgradient-alpha}.
\State Update the ellipsoid:
(a) $\widetilde{\bm{\eta}}:=\frac{1}{\sqrt{\bm{\eta}^T\mathbf{A}\bm{\eta}}}\bm{\eta}$;
(b) $\bm{\alpha}:=\bm{\alpha}-\frac{1}{3}\mathbf{A}\widetilde{\bm{\eta}}$;
(c) $\mathbf{A}:=\frac{4}{3}\left(\mathbf{A}{-}\frac{2}{3}\mathbf{A}\widetilde{\bm{\eta}}
\widetilde{\bm{\eta}}^T\mathbf{A}\right)$.
\Until $\bm{\alpha}$ converges to $\bm{\alpha}^{\star}(\bm{\lambda})$.
\State \textbf{Output} the optimal dual variable $\bm{\alpha}^{\star}(\bm{\lambda})$ for given $\bm{\lambda}$.
\end{algorithmic}
\end{algorithm}
\vspace{-4pt}
\begin{equation} \label{eq:subgradient-lambda-def}
G_\text{MA}(\bm{\lambda})\leq G_\text{MA}(\bm{\lambda}^0)+(\bm{\lambda}-\bm{\lambda}^0)^T\bm{\xi}(\bm{\lambda}^0),
\quad\forall~\bm{\lambda}\In\bm{\Lambda}_1\,{\textstyle \bigcup}\,\bm{\Lambda}_2.
\end{equation}
Suppose that $\bm{\lambda}^{\star}$ is an optimal solution to \eqref{eq:lambda-max1}, i.e., $G_\text{MA}(\bm{\lambda}^{\star})\geq G_\text{MA}(\bm{\lambda}^0)$. Then, by \eqref{eq:subgradient-lambda-def}, we must have
\vspace{-9pt}
\begin{equation} \label{eq:subgradient-lambda-optimal}
(\bm{\lambda}^{\star}-\bm{\lambda}^0)^T\bm{\xi}(\bm{\lambda}^0)\geq0
\end{equation}
for the optimal dual point $\bm{\lambda}^{\star}$.
In other words, if a dual point $\bm{\lambda}$ satisfies $(\bm{\lambda}-\bm{\lambda}^0)^T\bm{\xi}(\bm{\lambda}^0)\,{<}\,0$, then $\bm{\lambda}$ cannot be an optimal solution to problem \eqref{eq:lambda-max1}. Due to this and \eqref{eq:subgradient-lambda}, we establish the following proposition:

\begin{Proposition} \label{prop3}
Let $r_k~(k{=}1,2,3)$ denote the values of the terms used in the subgradient $\bm{\xi}(\bm{\lambda})$ in \eqref{eq:subgradient-lambda} with $\bm{\lambda}=\bm{\lambda}^0$.
If $r_3\geq r_1$, then $\bm{\lambda}^{\star}\In\bm{\Lambda}_1$. If $r_3\geq r_2$, then $\bm{\lambda}^{\star}\In\bm{\Lambda}_2$. If both $r_3\geq r_1$ and $r_3\geq r_2$, then
$\bm{\lambda}^{\star}=\bm{\lambda}^0=(0,0,1)$.
\end{Proposition}
\noindent\emph{Proof:} See Appendix~\ref{sec:proof-prop3}. \hfill{$\blacksquare$}

The procedure for solving \eqref{eq:lambda-max1} is given as follows: First, we utilize the preceding testing method stated in Proposition~\ref{prop3} to determine whether $\bm{\lambda}^{\star}\In\bm{\Lambda}_1$ or $\bm{\lambda}^{\star}\In\bm{\Lambda}_2$. Then, we use the bisection method to find the optimal dual variable $\bm{\lambda}^{\star}$. If $\bm{\lambda}^{\star}{=}(1-\lambda_3^{\star},0,\lambda_3^{\star})\In\bm{\Lambda}_1$, the directional subgradient $\zeta({\lambda}_3)$ of $G_\text{MA}(\bm{\lambda})$ along the direction of $\bm{\Lambda}_1$ is determined by
\begin{equation} \label{eq:subgradient-lambda3-1}
\zeta(\lambda_3)= \bm{\xi}(\bm{\lambda})^T\frac{\partial{\bm{\lambda}}}{\partial{\lambda_3}}=\bm{\xi}(\bm{\lambda})^T (-1,0,1)=r_1-r_3;
\end{equation}
otherwise, if $\bm{\lambda}^{\star}=(0,1-\lambda_3^{\star},\lambda_3^{\star})\In\bm{\Lambda}_2$, the directional subgradient $\zeta({\lambda}_3)$ along the direction of $\bm{\Lambda}_2$ is determined by
\begin{algorithm}
\caption{Proposed duality-based algorithm for solving subproblem \eqref{eq:MA-problem}} \label{alg2}
\begin{algorithmic}[1]  
\State \textbf{Input} CSI $\{g_{1n},g_{2n}\}_{n=1}^{N}$, average powers $\{P_1,P_2\}$, rate ratio $\rho$, and time proportion $t$.
\State Check whether $\bm{\lambda}^{\star}\In\bm{\Lambda}_1$ or $\bm{\lambda}^{\star}\In\bm{\Lambda}_2$ by Proposition~\ref{prop3}. If $\bm{\lambda}^{\star}=\bm{\lambda}^0=(0,0,1)$, go to Step~\ref{alg2-output}; otherwise, find $\bm{\lambda}^{\star}$ by the bisection method in Steps~\ref{alg2-bisection}-\ref{alg2-bisection1}.
\State Initialize $\lambda_{3,\min}=0,\,\lambda_{3,\max}=1$.\label{alg2-bisection}
\Repeat
\State Update $\lambda_3=\frac{1}{2}(\lambda_{3,\min}+\lambda_{3,\max})$.
\State Derive $\bm{\alpha}^{\star}(\bm{\lambda})$ for the inner-level dual optimization problem \eqref{eq:alpha-max} by Algorithm~\ref{alg1}.
\State Compute the subgradient $\bm{\xi}(\bm{\lambda})$ by \eqref{eq:subgradient-lambda} and the subgradient $\zeta(\lambda_3)$ by either \eqref{eq:subgradient-lambda3-1} or \eqref{eq:subgradient-lambda3-2}.
\State If the subgradient $\zeta(\lambda_3)<0$, then update $\lambda_{3,\max}=\lambda_3$; else update $\lambda_{3,\min}=\lambda_3$.
\Until $\lambda_3$ converges.\label{alg2-bisection1}
\State Obtain the optimal $\{\bm{p}_1^{\star},\bm{p}_2^{\star}\}$ by \eqref{eq:pa-case1}-\eqref{eq:pa-case3} and Remark~2. \label{alg2-output}
\State \textbf{Output} the optimal power allocation solution $\{\bm{p}_1^{\star},\bm{p}_2^{\star}\}$ and the optimal rate $R_{\text{MA}}^{\star}(t)$.
\end{algorithmic}
\end{algorithm}
\vspace{-9pt}
\begin{equation} \label{eq:subgradient-lambda3-2}
\zeta(\lambda_3)=\bm{\xi}(\bm{\lambda})^T(0,-1,1)=r_2-r_3.
\end{equation}
Since $G_\text{MA}(\bm{\lambda})$ is concave in $\bm{\lambda}$, it is also concave along the direction of $\bm{\Lambda}_1$ (or $\bm{\Lambda}_2$). Thus, $\zeta(\lambda_3)$ is monotonically non-increasing with respect to $\lambda_3$. Therefore, we can use the bisection method to search for the optimal solution $\lambda_3^{\star}$ to \eqref{eq:lambda-max}, which satisfies $\zeta({\lambda}_3^{\star})=0$, if $0<\lambda_3^{\star}<1$; $\zeta({\lambda}_3^{\star})\leq0$, if $\lambda_3^{\star}=0$; or $\zeta({\lambda}_3^{\star})\geq0$, if $\lambda_3^{\star}=1$.
The obtained algorithm for solving subproblem \eqref{eq:MA-problem} is summarized in Algorithm~\ref{alg2}.

\subsection{Lagrange Dual Optimization for Subproblem~\eqref{eq:BC-problem}} \label{sec:BC-problem}
The partial Lagrange dual function of subproblem \eqref{eq:BC-problem} is defined as
\vspace{-9pt}
\begin{equation} \label{eq:dual-function-BC}
D_\text{BC}(\lambda_4,\lambda_5,\alpha_3)\triangleq \min_{\bm{p}_R\succeq\bm{0},R_\text{BC}}L_\text{BC}\left(\bm{p}_R,R_\text{BC},\lambda_4,\lambda_5,\alpha_3\right),
\end{equation}
where $\lambda_4,\lambda_5$ and $\alpha_3$ are the nonnegative dual variables associated with two rate inequality constraints in \eqref{eq:rate-constr-BC} and one power inequality constraint \eqref{eq:power-constr-BC}, respectively, and
\vspace{-4pt}
\begin{eqnarray} \label{eq:Lagrangian-BC}
L_\text{BC}\left(\bm{p}_R,R_\text{BC},\lambda_4,\lambda_5,\alpha_3\right)
=-R_\text{BC}+\sum_{k=4}^5{\lambda_k\big[R_\text{BC}-r_k(\bm{p}_R)\big]}
+\alpha_3\Big(\sum_{n=1}^N{p_{Rn}}-P_R\Big).
\end{eqnarray}

Then, the corresponding dual optimization problem is defined as
\vspace{-9pt}
\begin{equation} \label{eq:dual-problem-BC}
\max_{\lambda_4\geq0,\lambda_5\geq0,\alpha_3\geq0}\;D_\text{BC}(\lambda_4,\lambda_5,\alpha_3).
\end{equation}

\begin{algorithm}
\caption{Proposed duality-based algorithm for solving subproblem \eqref{eq:BC-problem}} \label{alg3}
\begin{algorithmic}[1]  
\State \textbf{Input} CSI $\{\tilde{g}_{1n},\tilde{g}_{2n}\}_{n=1}^{N}$, average power $P_R$, rate ratio $\rho$, and time proportion $t$.
\State Initialize $\lambda_{5,\min}=0,\,\lambda_{5,\max}=1$.
\Repeat
\State Update $\lambda_5=\frac{1}{2}(\lambda_{5,\min}+\lambda_{5,\max})$ and initialize $\alpha_{3,\min},\,\alpha_{3,\max}$ with given $\lambda_5$.
\Repeat
\State Update $\alpha_3=\frac{1}{2}(\alpha_{3,\min}+\alpha_{3,\max})$.
\State Obtain the optimal $\bm{p}_R^{\star}$ by solving \eqref{eq:pa-BC0} at the dual point $(1-\lambda_5,\lambda_5,\alpha_3)$.
\State If $\sum_{n=1}^N{p_{Rn}^{\star}}<P_R$, then update $\alpha_{3,\max}=\alpha_3$; else update $\alpha_{3,\min}=\alpha_3$.
\Until $\alpha_3$ converges to $\alpha_3^{\star}(\lambda_5)$.
\State Obtain the optimal $\bm{p}_R^{\star}$ by solving \eqref{eq:pa-BC0} at the dual point $(1-\lambda_5,\lambda_5,\alpha_3^{\star}(\lambda_5))$.
\State If $r_4(\bm{p}_R^{\star})<r_5(\bm{p}_R^{\star})$, then update $\lambda_{5,\max}=\lambda_5$; else update $\lambda_{5,\min}=\lambda_5$.
\Until $\lambda_5$ converges.
\State Obtain the optimal $\bm{p}_R^{\star}$ by \eqref{eq:pa-BC0}.
\State \textbf{Output} the optimal power allocation solution $\bm{p}_R^{\star}$ and the optimal rate $R_{\text{BC}}^{\star}(t)$.
\end{algorithmic}
\end{algorithm}

The KKT conditions associated with \eqref{eq:dual-function-BC} are given by
\vspace{-4pt}
\begin{subequations} \label{eq:BC-KKT}
\begin{align}
&\frac{\partial{L_{\text{BC}}}}{\partial{p_{Rn}}}=
\alpha_3\!-\!\frac{(1-t){\tilde{g}_{2n}}\lambda_4}{(1-t+\tilde{g}_{2n}p_{Rn}^{\star})\ln2}
-\frac{(1-t){\tilde{g}_{1n}}\lambda_5}{\rho(1-t+\tilde{g}_{1n}p_{Rn}^{\star})\ln2}
\;\left\{\begin{array}{ll}
\geq0,&\;\text{if}\;p_{Rn}^{\star}=0; \\[2pt]
=0,&\;\text{if}\;p_{Rn}^{\star}>0,
\end{array}\right. \label{eq:BC-KKT1}\\
&\frac{\partial{L_{\text{BC}}}}{\partial{R_{\text{BC}}}}=
\lambda_4^{\star}+\lambda_5^{\star}-1=0. \label{eq:BC-KKT2}
\end{align}
\end{subequations}
If $p_{Rn}^{\star}>0$, then the equality in \eqref{eq:BC-KKT1} holds, and the optimal $p_{Rn}^{\star}$ can be shown to be the positive root $x$ of the following quadratic
equation \cite{Jitvan_TVT09}
\begin{equation}\label{eq:pa-BC0}
\frac{(1-t){\tilde{g}_{2n}}\lambda_4}{1-t+\tilde{g}_{2n}x}
+\frac{(1-t){\tilde{g}_{1n}}\lambda_5}{\rho(1-t+\tilde{g}_{1n}x)}=\alpha_3\ln2.
\end{equation}
If \eqref{eq:pa-BC0} has no positive root, then $p_{Rn}^{\star}=0$.
By \eqref{eq:BC-KKT2}, we have $\lambda_4^{\star}=1-\lambda_5^{\star}$. Thus, the optimal dual variables $(\alpha_3^{\star},\lambda_5^{\star})$ can be derived by a two-level bisection optimization method, and the obtained algorithm for solving subproblem \eqref{eq:BC-problem} is summarized in Algorithm~\ref{alg3}.
\ifreport
More details about the initialization of $\alpha_{3,\min}$ and $\alpha_{3,\max}$ in Algorithm~\ref{alg3} are given in Appendix~\ref{sec:alpha-bounds}.
\fi

\begin{algorithm}
\caption{Proposed resource allocation algorithm for solving problem \eqref{eq:original-problem}} \label{alg4}
\begin{algorithmic}[1]
\State \textbf{Input} CSI $\{g_{1n},g_{2n},\tilde{g}_{1n},\tilde{g}_{2n}\}_{n=1}^{N}$, average powers $\{P_1,P_2,P_R\}$, rate ratio $\rho$.
\Repeat
\State Solve the power allocation subproblems \eqref{eq:MA-problem} and \eqref{eq:BC-problem} by Algorithm~\ref{alg2} and Algorithm~\ref{alg3}, $\quad~~~~~~$respectively, where $t$ is a given parameter.
\State Update $t$ using the one-dimensional search method for problem \eqref{eq:phase-decomp-mu}.
\Until $t$ converges.
\State \textbf{Output} the optimal resource allocation $\{\bm{p}_1^{\star},\bm{p}_2^{\star},\bm{p}_R^{\star},t^{\star}\}$ and the optimal rate $R_{12}^{\star}$.
\end{algorithmic}
\end{algorithm}

As previously mentioned, after solving the power allocation subproblems \eqref{eq:MA-problem} and \eqref{eq:BC-problem}, problem \eqref{eq:phase-decomp-mu} can be solved by utilizing the efficient one-dimensional search methods in \cite[Chapter~8]{Bazaraa06}, thereby yielding Algorithm~\ref{alg4} for solving problem \eqref{eq:original-problem}.

\subsection{Computational Complexity Analysis}
The computational complexity of Algorithm~\ref{alg2} is given by $O\left(L(2)KNC_1\right)$, where
$L(n)=O\left(2(n{+}1)^2\ln(1/\epsilon)\right)$ is the number of iterations in the ellipsoid method for an $n$-variable nonsmooth optimization problem \cite[p.~155]{Nesterov04}, $\epsilon$ is the accuracy of the obtained solution, $K=O\left(\ln(1/\epsilon)\right)$ is the complexity (abbreviated for the complexity order) of one-dimensional search methods such as the bisection method, $C_1$ is the complexity for computing the closed-form power allocation solution \eqref{eq:pa-case1}-\eqref{eq:pa-case3}.
The computational complexity of Algorithm~\ref{alg3} is given by $O\left(K^2NC_2\right)$, where $C_2$ is the complexity for computing the closed-form power allocation solution to \eqref{eq:pa-BC0}.
Therefore, the overall computational complexity of the proposed resource allocation algorithm (Algorithm~\ref{alg4}) is given by $O\left(L(2)K^2NC_1+K^3NC_2\right)$.

The complexity of the resource allocation algorithm for the per-subcarrier DF relay strategy in \cite{Jitvan_TVT09} is given by  $O(L(2)L(3)KN(I+C_2))$, where $I$ is the complexity of the iterative power allocation algorithm in Eq.~(26) and (27) in \cite{Jitvan_TVT09},
$C_2$ is the complexity of the closed-form power allocation solution in Eq.~(28) in \cite{Jitvan_TVT09}. The complexities of the algorithm in \cite{Jitvan_TVT09} and Algorithm~\ref{alg4} both grow linearly with the number of subcarriers $N$, and therefore they are quite appropriate for practically large values of $N$.
In addition, the computational complexity of the iterative power allocation algorithm $I$ is much larger than that of the closed-form power allocation solution $C_1$. The computational complexity of the ellipsoid method $L(3)$ is much larger than that of one-dimensional search methods $K$. Therefore, the computational complexity of Algorithm~\ref{alg4} is much smaller than that of the resource allocation algorithm in \cite{Jitvan_TVT09}.

\section{Asymptotic Performance Analysis} \label{sec:bound-analysis}
In this section, we analyze the asymptotic rate regions of different relay strategies for two-way OFDM channels, including both the per-subcarrier and the proposed multi-subcarrier DF relay strategies, the AF relay strategy, and the cut-set outer bound, in order to compare their achievable rate regions in both low and high SNR regimes and their respective performance merits.

The cut-set outer bound for the capacity region of the two-way OFDM relay channels \eqref{eq:MA-channel}-\eqref{eq:BC-channel} is obtained by removing the sum-rate constraints in \eqref{eq:prop1}, which is given by \cite{Kim_TIT08}
\vspace{-4pt}
{\setlength\arraycolsep{0pt}
\begin{eqnarray}\label{eq:cutset}
\mathcal{R}_\text{out}(\bm P,\mathcal{G})=\Bigg\{
&&(R_{12},R_{21})\In\Rset_+^2~\bigg|\nonumber\\[-4pt]
&&R_{12}\leq\min\!\bigg\{\sum_{n=1}^Nt{\log_2\!\Big(1{+}\frac{g_{1n}p_{1n}}{t}\Big)},
\sum_{n=1}^N(1-t){\log_2\!\Big(1{+}\frac{\tilde{g}_{2n}p_{Rn}}{1-t}\Big)}\bigg\},\nonumber\\[-4pt]
&&R_{21}\leq\min\!\bigg\{\sum_{n=1}^Nt{\log_2\!\Big(1\!+\!\frac{g_{2n}p_{2n}}{t}\Big)},
\sum_{n=1}^N(1-t){\log_2\!\Big(1{+}\frac{\tilde{g}_{1n}p_{Rn}}{1-t}\Big)}\bigg\},\nonumber\\[-4pt]
&&0< t< 1,\,\sum_{n=1}^N{p_{in}}\leq P_i,\,p_{in}\geq0,\,i=1,2,R,\,n=1,\ldots,N
\Bigg\}.
\vspace{-4pt}
\end{eqnarray}}

The achievable rate region for the AF relay strategy is given by \cite{Jitvan_TVT09}
{\setlength\arraycolsep{1pt}
\begin{eqnarray}\label{eq:AF-region}
\mathcal{R}_\text{AF}(\bm P,\mathcal{G})=
\Bigg\{(R_{12},R_{21})\In\Rset_+^2~\bigg|
&&R_{12}\leq\sum_{n=1}^N\!\frac{1}{2}\log_2\!\Big(1\!+\!\frac{2p_{1n}g_{1n}\tilde{g}_{2n}a_n}
{1+\tilde{g}_{2n}a_n}\Big),\nonumber\\[-4pt]
&&R_{21}\leq\sum_{n=1}^N\!\frac{1}{2}\log_2\!\Big(1\!+\!\frac{2p_{2n}g_{2n}\tilde{g}_{1n}a_n}
{1+\tilde{g}_{1n}a_n}\Big),\nonumber\\[-4pt]
&&\sum_{n=1}^N{p_{in}}\leq P_i,\,p_{in}\geq0,\,i=1,2,R,\,n=1,\ldots,N\Bigg\},
\vspace{-4pt}
\end{eqnarray}}
$\!\!$where $a_n=\frac{p_{Rn}}{p_{1n}g_{1n}+p_{2n}g_{2n}+1}$ is the amplification factor of the relay node in subcarrier $n$ and the time proportion $t$ is fixed to be $0.5$ due to the incompressible nature of the AF relay strategy.

Suppose that $\bar{\bm P}\triangleq (\bar{P}_1,\bar{P}_2,\bar{P}_R)$ is a column vector constituted by nominal values of  $P_1$, $P_2$ and $P_R$. Then the available transmission powers of the three nodes can be expressed as
\vspace{-4pt}
\begin{equation}\label{eq:scale}
\bm{P}=x\bar{\bm P},
\end{equation}
where $x$ is a positive scalar variable. Note that the average SNRs of all the wireless links are proportional to $x$, and so we will analyze the achievable rate regions of the two-way relay strategies under consideration for small $x$ and large $x$ instead.

\subsection{Low SNR Regime}
In the low SNR region (small $x$), the function $\log_2(1+ax)$ with $a>0$ can be expressed as
\vspace{-4pt}
\begin{eqnarray}\label{eq:taylor-low}
\log_2(1+ax)=\frac{a}{\ln2}\,x+O(x^2).
\end{eqnarray}
Using \eqref{eq:taylor-low}, we can show the following proposition:

\begin{Proposition}\label{prop5}
For sufficiently small $x\,{>}\,0$ and any rate pair $(R_{12},R_{21})\In\mathcal{R}_\text{out}(x\bar{\bm P},\mathcal{G})$, there exists some $(\hat{R}_{12}, \hat{R}_{21})\In\mathcal{R}_\text{DF}(x\bar{\bm P},\mathcal{G})$ such that $R_{12}=\hat{R}_{12}+O(x^{2b})$ and $R_{21}=\hat{R}_{21}+O(x^{2b})$ for $b\geq1$. The regions $\mathcal{R}_\text{DF}(x\bar{\bm P},\mathcal{G})$ and $\mathcal{R}_\text{out}(x\bar{\bm P},\mathcal{G})$ tend to be the same as $x\,{\rightarrow}\,0$.
\end{Proposition}
Therefore, the multi-subcarrier DF relay strategy tends to achieve the capacity region of two-way OFDM relay channels \eqref{eq:MA-channel}-\eqref{eq:BC-channel} as $x\,{\rightarrow}\,0$.
\ifreport
\begin{proof}
See Appendix~\ref{sec:proof-prop5}.
\end{proof}
\fi
On the other hand, it can be easily shown that the achievable rate region of the AF relay strategy will deflate in a much faster speed than the other two-way strategies for small $x$, due to the noise amplification and propagation effects.

\subsection{High SNR Regime}
In the high SNR region (large $x$), the function $\log_2(1+ax)$ with $a>0$ satisfies
\vspace{-4pt}
\begin{eqnarray}\label{eq:taylor-high}
\log_2(1+ax)=\log_2(ax)+O(1/x)=\log_2(x)+O(1).
\end{eqnarray}

Let us define the multiplexing gain region of the multi-subcarrier DF relay strategy \cite{Zheng_TIT03}:
\begin{eqnarray}\label{eq:multiplexing-DF}
r_\text{DF}\triangleq
\lim_{x\to\infty} \frac{\mathcal{R}_\text{DF}(x\bar{\bm P},\mathcal{G})}{\log_2(x)}.
\end{eqnarray}
Using \eqref{eq:taylor-high}, we can prove the following proposition:
\begin{Proposition}\label{prop7}
The multiplexing gain region of the multi-subcarrier DF relay strategy is given by
\vspace{-15pt}
\begin{eqnarray}\label{eq:multiplexing}
r_\text{DF}=\Big\{(r_{12},r_{21})~\Big|~r_{12}+2r_{21}\leq N,\,2r_{12}+r_{21}\leq N,\,r_{12},r_{21}\geq0\Big\}.
\end{eqnarray}
\end{Proposition}
\begin{proof}
To prove \eqref{eq:multiplexing}, it is sufficient to find two rate regions ${\mathcal{R}}_1(x\bar{\bm P},\mathcal{G})$ and
${\mathcal{R}}_2(x\bar{\bm P},\mathcal{G})$, such that
${\mathcal{R}}_1(x\bar{\bm P},\mathcal{G})\subset{\mathcal{R}}_\text{DF}(x\bar{\bm P},\mathcal{G})\subset{\mathcal{R}}_2(x\bar{\bm P},\mathcal{G})$, and the corresponding multiplexing gain regions of ${\mathcal{R}}_1(x\bar{\bm P},\mathcal{G})$ and
${\mathcal{R}}_2(x\bar{\bm P},\mathcal{G})$ are both given by \eqref{eq:multiplexing}.
\ifreport
The detailed proof is given in Appendix~\ref{sec:proof-prop7}.
\fi
\end{proof}

Actually, the multiplexing gain region $r_\text{DF}$ given by \eqref{eq:multiplexing} depends on the time proportion allocation but not upon the power allocation,
\ifreport
which can be observed from the proof of Proposition~\ref{prop7} in Appendix~\ref{sec:proof-prop7}.
\fi
For instance, the simple equal power allocation scheme can achieve this multiplexing gain region, and thereby the achievable rate region gap between this power allocation scheme and the optimal power allocation scheme asymptotically converges to a constant region gap for sufficiently large $x$.

Following similar ideas for the proof of Proposition~\ref{prop7}, one can derive the multiplexing gains for the per-subcarrier DF relay strategy, the AF relay strategy, and the cut-set outer bound as stated in the following proposition (with the proof omitted):
\begin{Proposition}\label{prop8}
Let $r_\text{p,DF}$, $r_\text{AF}$, and $r_\text{out}$ denote the multiplexing gain regions of the per-subcarrier DF relay strategy, the AF relay strategy and the cut-set outer bound, respectively. Then $r_\text{p,DF}=r_\text{DF}$ (given by \eqref{eq:multiplexing}) and $r_\text{AF}=r_\text{out}=\Big\{(r_{12},r_{21})~\Big|~r_{12}\leq \frac{N}{2},\,r_{21}\leq \frac{N}{2},\,r_{12},r_{21}\geq0\Big\}$.
\end{Proposition}

Proposition~\ref{prop8} implies that the AF relay strategy can achieve the multiplexing gain region of the two-way OFDM relay channels, while the performance of both DF relay strategies is worse than that of the AF relay strategy in the high SNR regime. An illustrative example for these analytical results is given in Fig.~\ref{fig:region-compare-multiplxing-gain}.
To the best of our knowledge, the multiplexing gain region of the cut-set outer bound was derived in \cite{Gunduz_Asilomar08}, while the multiplexing gain regions of the DF and AF relay strategies have not been reported in the open literature before. All the analytical results as presented in Propositions~\ref{prop5}-\ref{prop8} will be confirmed by our numerical results in the next section.

\begin{figure}[t]
    \centering
    \scalebox{0.6}{\includegraphics*[83,123][405,441]{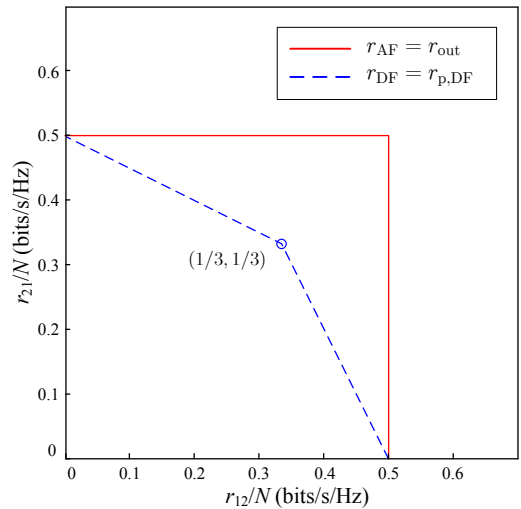}}
    \caption{Comparison of multiplexing gain regions of different two-way OFDM relay strategies in the high SNR regime.} \label{fig:region-compare-multiplxing-gain}
\end{figure}

\section{Numerical Results} \label{sec:simulation}
We now provide some numerical results to compare the performance of different
two-way OFDM relay strategies under optimal resource allocation. The wireless channels are generated by using $M=4$ independent Rayleigh distributed time-domain taps. The number of subcarriers in the OFDM channel is $N=16$. We assume that the wireless channels between $T_i\,(i\In\{1,2\})$ and $T_R$ are reciprocal, i.e., $g_{in}=\tilde{g}_{in}$, for all $i=1,2,\;n=1,\ldots,16$. The maximum average transmission powers for all the nodes are assumed to be the same, i.e., $P_1=P_2=P_R=P$. Therefore, the average SNR of the wireless links between $T_i\,(i\In\{1,2\})$ and $T_R$ is given by $\text{SNR}_i=\mathbb{E}[g_{in}]\frac{P}{N}$.

We consider the following two-way relay strategies in our numerical comparisons: the multi-subcarrier DF relay strategy proposed in Lemma~\ref{prop1}, the per-subcarrier DF relay strategy \cite{Jitvan_TVT09}, the AF relay strategy \cite{Jitvan_TVT09}, and the cut-set outer bound \cite{Kim_TIT08}. The associated rate regions of these strategies are given by ${\mathcal{R}}_\text{DF}(\bm P,\mathcal{G})$ in \eqref{eq:prop1}, ${\mathcal{R}}_\text{p,DF}(\bm P,\mathcal{G})$ in \eqref{eq:region-psc}, ${\mathcal{R}}_\text{AF}(\bm P,\mathcal{G})$ in \eqref{eq:AF-region}, and ${\mathcal{R}}_\text{out}(\bm P,\mathcal{G})$ in \eqref{eq:cutset}, respectively. The optimal resource allocation of the multi-subcarrier DF relay strategy is obtained by Algorithm~\ref{alg4}, the optimal resource allocation of the per-subcarrier DF relay strategy and the AF relay strategy are carried out based on the power allocation algorithms proposed in \cite{Jitvan_TVT09}, and the optimal resource allocation of the cut-set outer bound is obtained by a simpler version of Algorithm~\ref{alg4}.

\begin{figure*}[t]
    \centering
    \subfigure[$\text{SNR}\,{=}\,0\,$dB.]{
    \scalebox{0.6}{\includegraphics*[78,476][418,816]{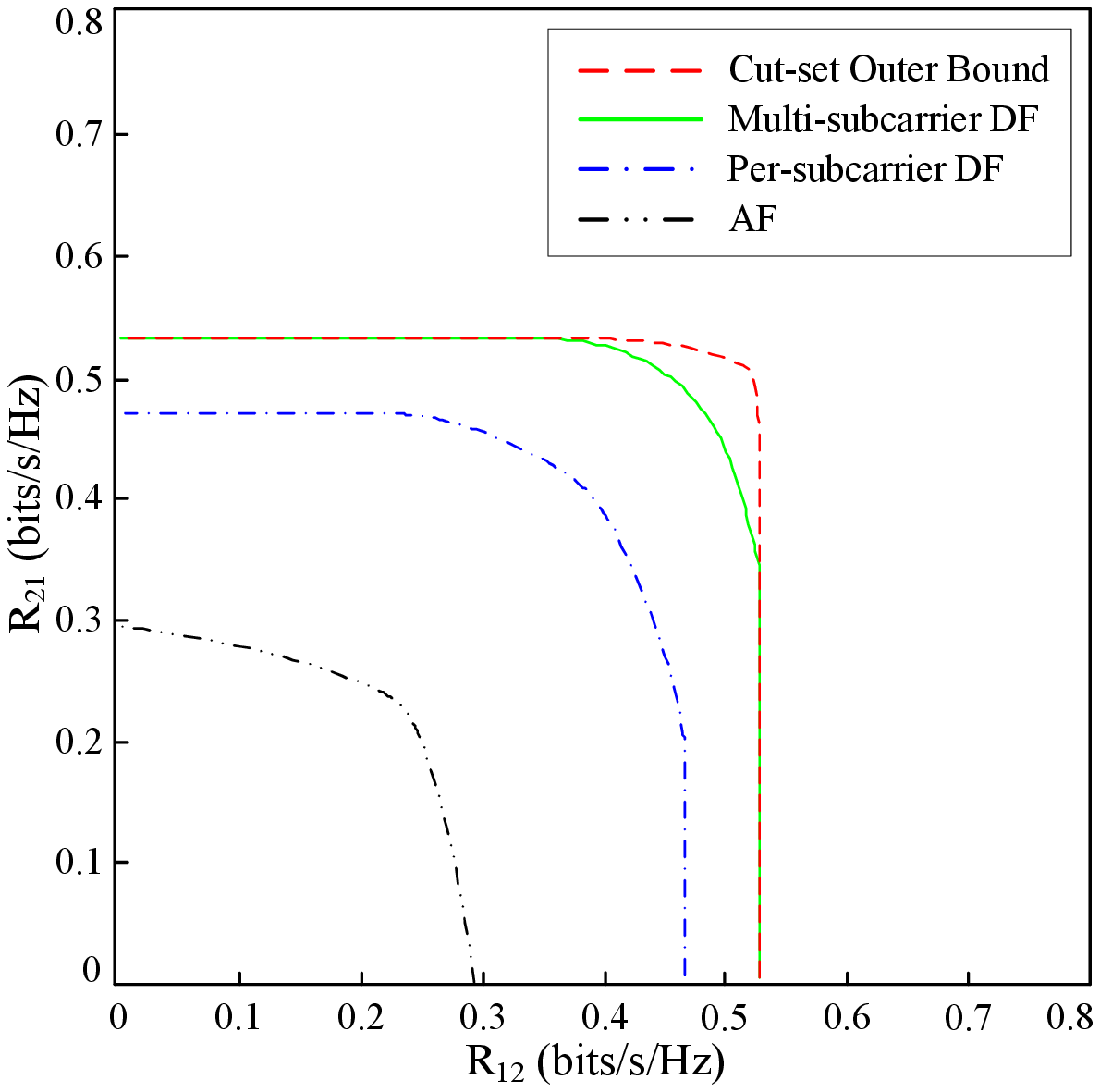}} \label{fig:region-compare-0dB}}
    \subfigure[$\text{SNR}\,{=}\,10\,$dB.]{
    \scalebox{0.6}{\includegraphics*[78,110][418,448]{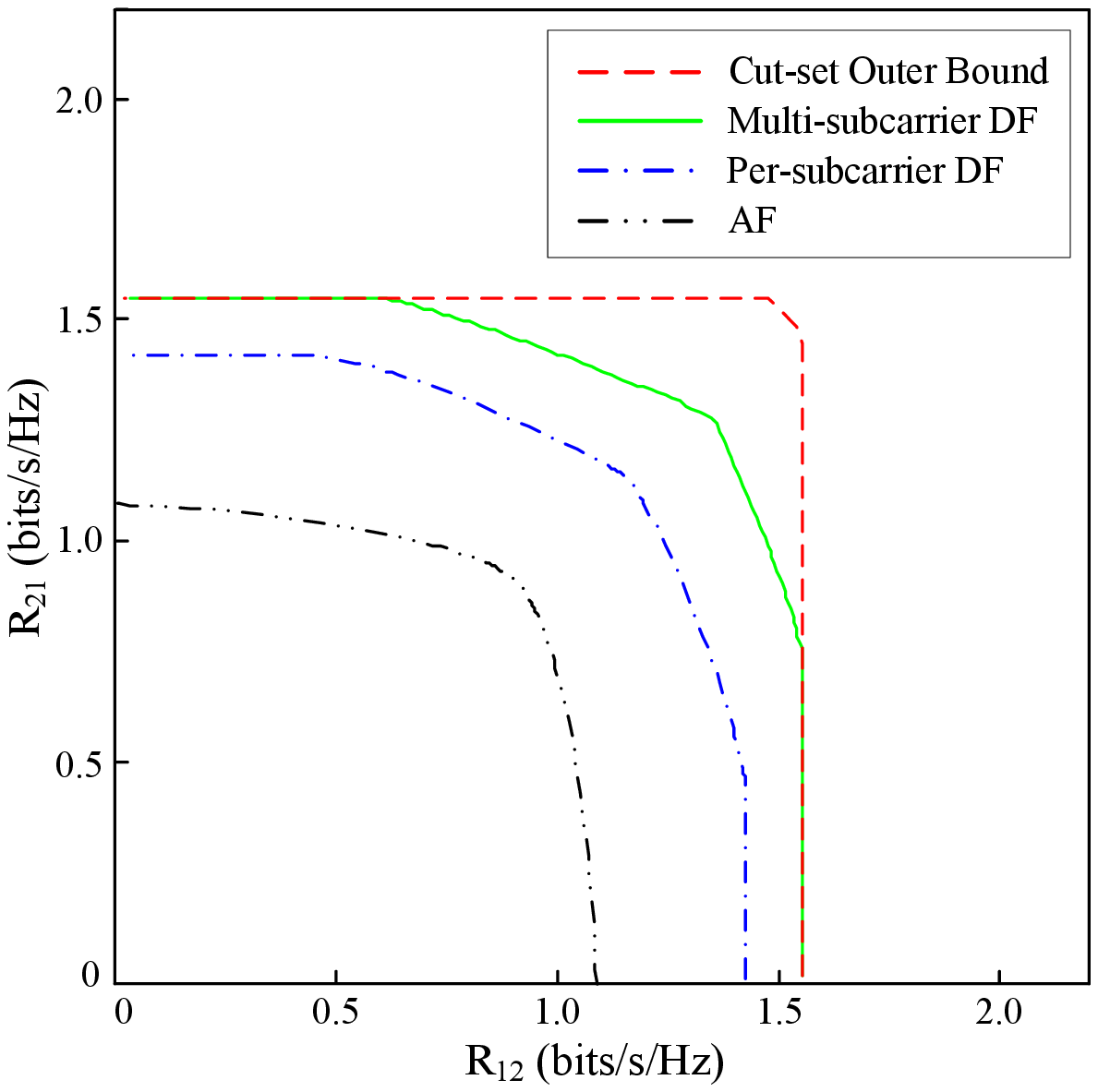}} \label{fig:region-compare-10dB}}
    \vspace{20pt}
    \subfigure[$\text{SNR}\,{=}\,20\,$dB.]{
    \scalebox{0.6}{\includegraphics*[78,476][418,816]{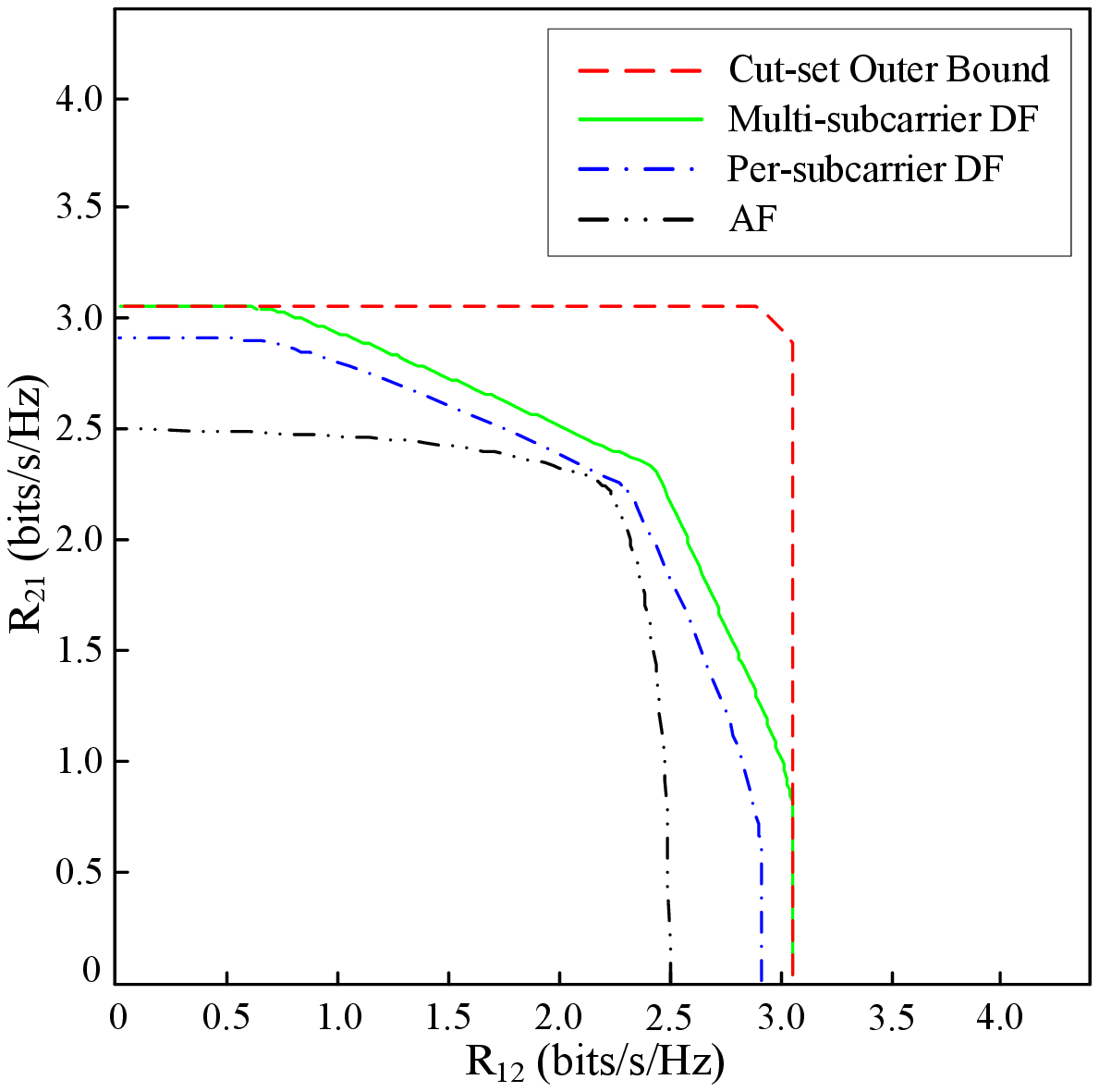}} \label{fig:region-compare-20dB}}
    \subfigure[$\text{SNR}\,{=}\,30\,$dB.]{
    \scalebox{0.6}{\includegraphics*[78,110][418,448]{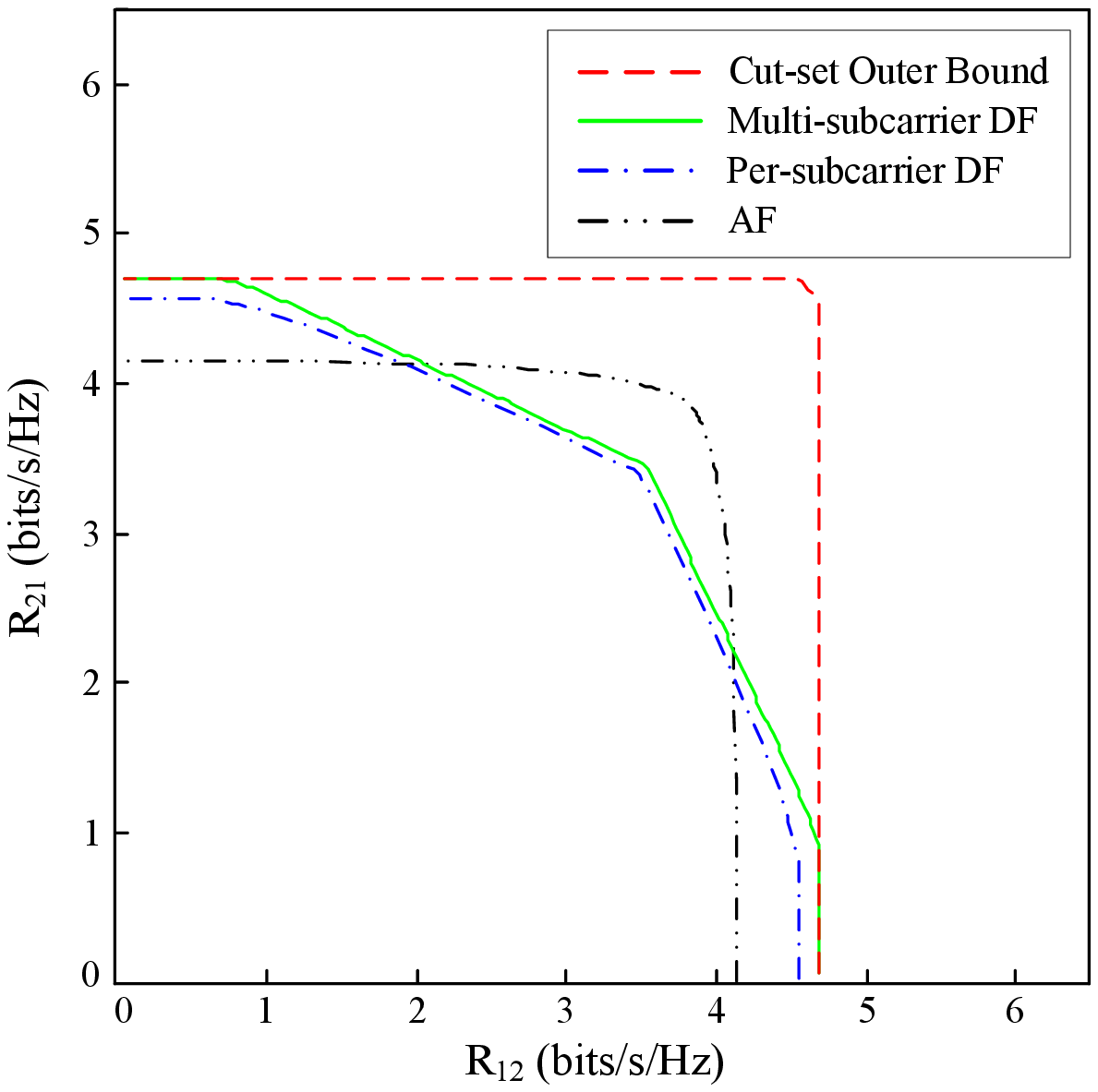}} \label{fig:region-compare-30dB}}
    \caption{Achievable rate regions of four two-way OFDM relay strategies for four symmetric SNR scenarios (i.e., $\text{SNR}_1\,{=}\,\text{SNR}_2\,{=}\,\text{SNR}$), including (a) $\text{SNR}\,{=}\,0\,$dB, (b) $\text{SNR}\,{=}\,10\,$dB, (c) $\text{SNR}\,{=}\,20\,$dB, and (d) $\text{SNR}\,{=}\,30\,$dB.} \label{fig:region-compare}
    \vspace{-20pt}
\end{figure*}

Figures~\ref{fig:region-compare-0dB}-\ref{fig:region-compare-30dB} provide the rate regions of these relay strategies for four symmetric SNR scenarios with $\text{SNR}_1\,{=}\,\text{SNR}_2\,{=}\,\text{SNR}\,{=}\,0,10,20,30\,$dB, respectively.
Some observations from these figures are worth mentioning: First, the achievable rate region of the multi-subcarrier DF relay strategy is always larger than that of the per-subcarrier DF relay strategy. Second, as the SNR decreases, the achievable rate region of the multi-subcarrier DF relay strategy tends to reach the cut-set outer bound. Third, the achievable rate region of the AF relay strategy grows with SNR, but it is still a subset of those of the DF relay strategies for $\text{SNR}\,{\leq}\,20\,$dB; this is no longer true for $\text{SNR}\,{=}\,30\,$dB. Finally, in the high SNR region, the rate regions of these strategies tend to be dominated by the multiplexing gain region, thereby consistent with Propositions~\ref{prop7} and \ref{prop8}. To be more specific, the shape of the outer bound tends to be a rectangle depending on the SNR. The achievable rate region of the AF strategy is closer to the outer bound for the higher SNR, but that of the two DF strategies are not. However, for the low SNR, only the proposed multi-subcarrier DF strategy can approach the outer bound.

\ifreport
Figure~\ref{fig:t-rho-20dB} provides the optimal channel resource allocation result $t^\star$ versus the rate ratio $\rho\,{=}\,R_{21}/R_{12}$ of the multi-subcarrier DF strategy and the cut-set outer bound for $\text{SNR}_1=\text{SNR}_2=20\,$dB. When $0<\rho<0.2$ or $\rho>5$, the optimal $t^\star$ of the multi-subcarrier DF strategy and the cut-set outer bound are the same; when $0.2<\rho<5$, the optimal $t^\star$ of the multi-subcarrier DF strategy is larger than that of the cut-set outer bound due to the additional sum-rate constraint, and the maximal $t^\star$ is achieved at $\rho=1$.
\else
\fi
Figure~\ref{fig:region-compare-CVX-20dB} shows the achievable rate region of the multi-subcarrier DF relay strategy obtained by solving problem \eqref{eq:original-problem} using CVX, and that obtained by using Algorithm~\ref{alg4}, justifying that they yield the same numerical results as expected.

\ifreport
\begin{figure}[t]
    \centering
    \scalebox{0.6}{\includegraphics*[53,117][465,443]{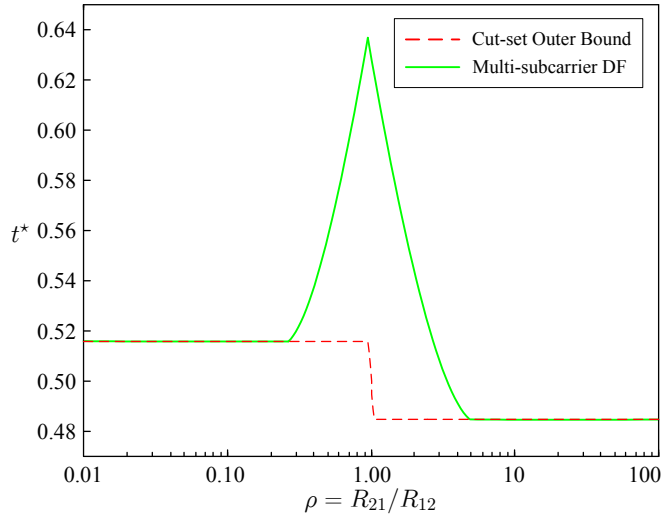}}
    \caption{The optimal time proportion of the multiple-access phase $t^\star$ versus the rate ratio $\rho\,{=}\,R_{21}/R_{12}$ of the multi-subcarrier DF strategy and the cut-set outer bound for $\text{SNR}_1\,{=}\,\text{SNR}_2\,{=}\,20\,$dB.} \label{fig:t-rho-20dB}
    \vspace{-20pt}
\end{figure}
\else
\fi

Figure~\ref{fig:region-compare-10dB-5dB} and \ref{fig:region-compare-30dB-5dB} illustrate the rate regions of these relay strategies for two asymmetric SNR scenarios, including $(\text{SNR}_1,\text{SNR}_2)\,{=}\,(10\,\text{dB},5\,\text{dB})$ and $(\text{SNR}_1,\text{SNR}_2)\,{=}\,(30\,\text{dB},5\,\text{dB})$.
Similar observations from Figure~\ref{fig:region-compare} can be seen in Figure~\ref{fig:region-compare-asym} as well.

Finally, Figure~\ref{fig:rate-msc-psc} shows some results (the achievable rate versus average SNR) of these relay strategies for the symmetric SNR symmetric rate scenario, i.e., $\text{SNR}_1\,{=}\,\text{SNR}_2\,{=}\,\text{SNR}$ and $R_{12}=R_{21}$. The numerical results in Fig.~\ref{fig:rate-msc-psc} were obtained by averaging over 500 fading channel realizations. One can see from this figure that, in the low SNR regime, the multi-subcarrier DF relay strategy tends to have the same performance as the cut-set outer bound, and that the multi-subcarrier DF relay strategy performs better than the AF relay strategy in the low to moderate SNR regime, i.e., $\text{SNR}\,{\leq}\,24\,$dB. Moreover, the multi-subcarrier DF relay strategy with the optimal resource allocation performs better than with the equal power allocation and the optimal $t^\star$ used; it also outperforms the per-subcarrier DF strategy.

\begin{figure}[t]
    \centering
    \scalebox{0.6}{\includegraphics*[78,476][398,796]{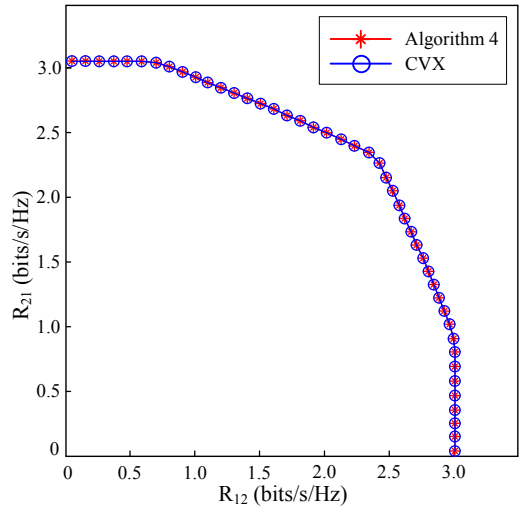}}
    \caption{Achievable rate region of the multi-subcarrier DF relay strategy obtained by using Algorithm~\ref{alg4} and that obtained by solving problem \eqref{eq:original-problem} using the convex solver CVX for $\text{SNR}_1\,{=}\,\text{SNR}_2\,{=}\,20\,$dB.} \label{fig:region-compare-CVX-20dB}
    \vspace{-20pt}
\end{figure}

By Proposition~\ref{prop8}, in the high SNR regime, the multiplexing gains of the AF relay strategy and the cut-set outer bound are the same; the multiplexing gains of the two DF relay strategies are also the same; the multiplexing gain of the AF relay strategy is larger than that of the DF relay strategy (implying better performance for the former than the latter for sufficiently high SNR); both the equal power allocation and the optimal power allocation for the multi-subcarrier DF strategy achieve the same multiplexing gain, and the rate gap between them tends to a constant value as SNR increases. All these analytical results have been substantiated by the numerical results shown in Figures~\ref{fig:region-compare}-\ref{fig:rate-msc-psc}.

\section{Conclusion} \label{sec:conclusion}
We have analytically shown that the widely studied per-subcarrier DF relay strategy is only a suboptimal DF relay strategy for two-way OFDM relay channels in terms of achievable rate region as the direct link between the two terminal nodes is too weak to be used for data transmission. We have presented a multi-subcarrier DF relay strategy that can achieve a larger rate region than the per-subcarrier DF relay strategy. Although the optimal resource allocation for the proposed multi-subcarrier DF relay strategy can be formulated as a convex optimization problem which can be solved by off-the-self convex solvers, we have presented a computationally efficient algorithm (Algorithm~\ref{alg4}) for obtaining the optimal resource allocation together with its complexity analysis. Then we have presented an analysis of asymptotic performance for the above two DF strategies, the AF strategy, and the cut-set outer bound, justifying that the proposed multi-subcarrier DF relay strategy is suitable for the low to moderate SNR regime, while the AF strategy is suitable for the high SNR regime (as stated in Propositions~\ref{prop5} to \ref{prop8}). Some numerical results have been presented to demonstrate all the analytical results and the efficacy of the proposed optimal resource allocation algorithm.

\ifreport
\begin{figure*}[t]
    \centering
    \subfigure[$\text{SNR}_1\,{=}\,10\,\text{dB},\,\text{SNR}_2\,{=}\,5\,\text{dB}$.]{
    \scalebox{0.6}{\includegraphics*[78,476][418,816]{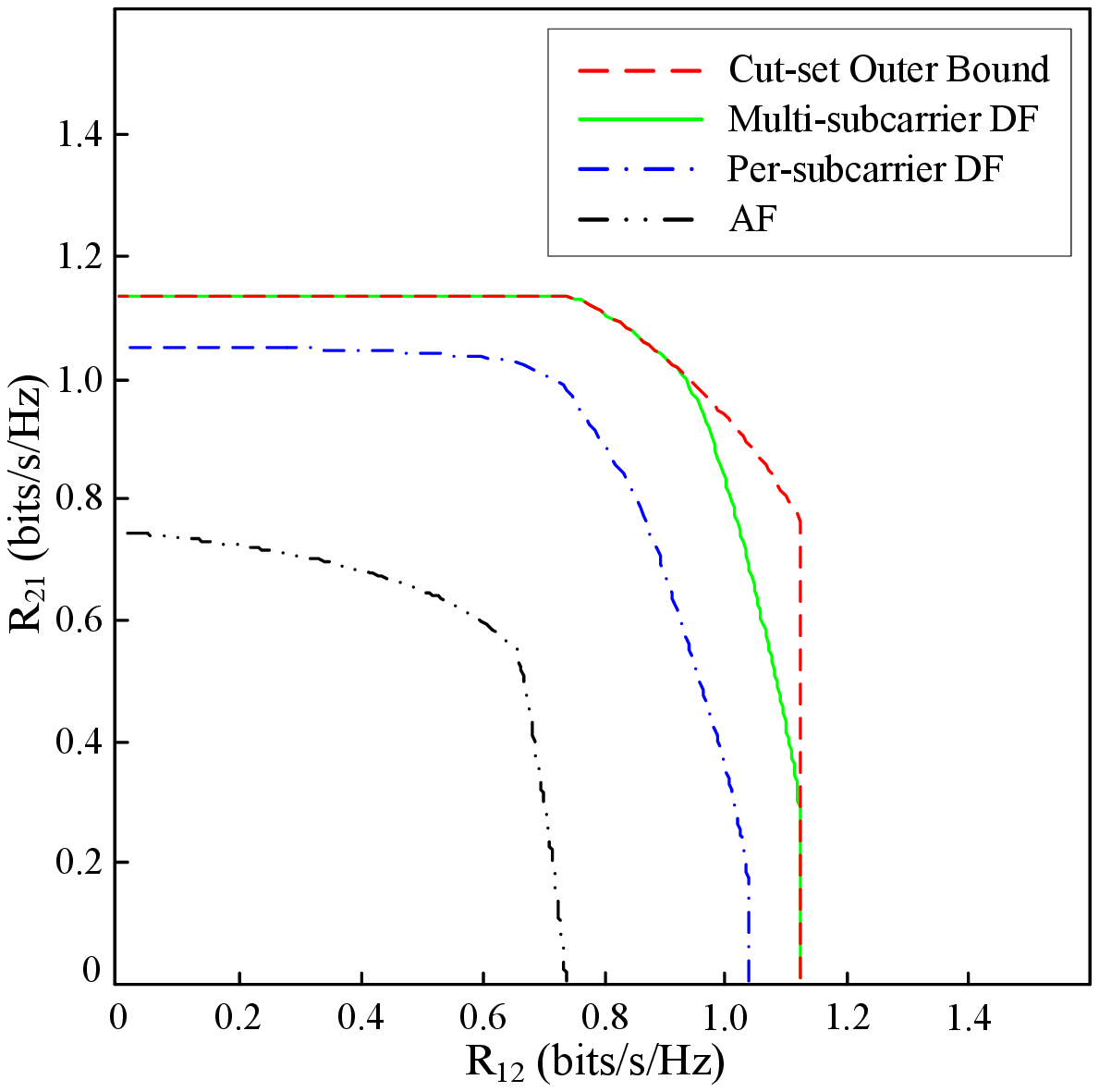}} \label{fig:region-compare-10dB-5dB}}
    \subfigure[$\text{SNR}_1\,{=}\,30\,\text{dB},\,\text{SNR}_2\,{=}\,5\,\text{dB}$.]{
    \scalebox{0.6}{\includegraphics*[78,476][418,816]{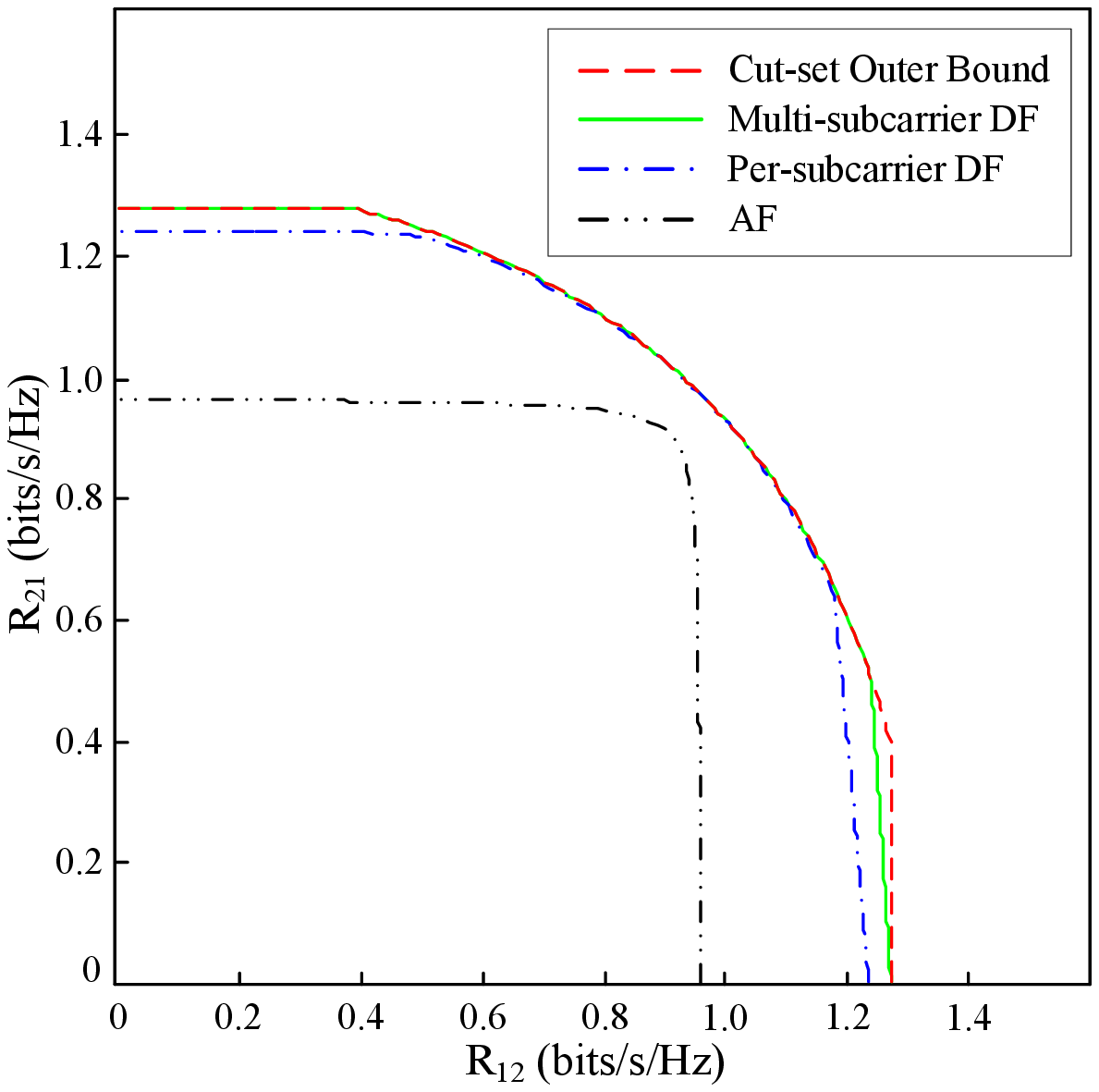}} \label{fig:region-compare-30dB-5dB}}
    \caption{Achievable rate regions of four two-way OFDM relay strategies for two asymmetric SNR scenarios, including (a) $\text{SNR}_1\,{=}\,10\,\text{dB},\,\text{SNR}_2\,{=}\,5\,\text{dB}$, and (b) $\text{SNR}_1\,{=}\,30\,\text{dB},\,\text{SNR}_2\,{=}\,5\,\text{dB}$.} \label{fig:region-compare-asym}
    \vspace{-20pt}
\end{figure*}

\begin{figure}[t]
    \centering
    \scalebox{0.65}{\includegraphics*[42,424][528,820]{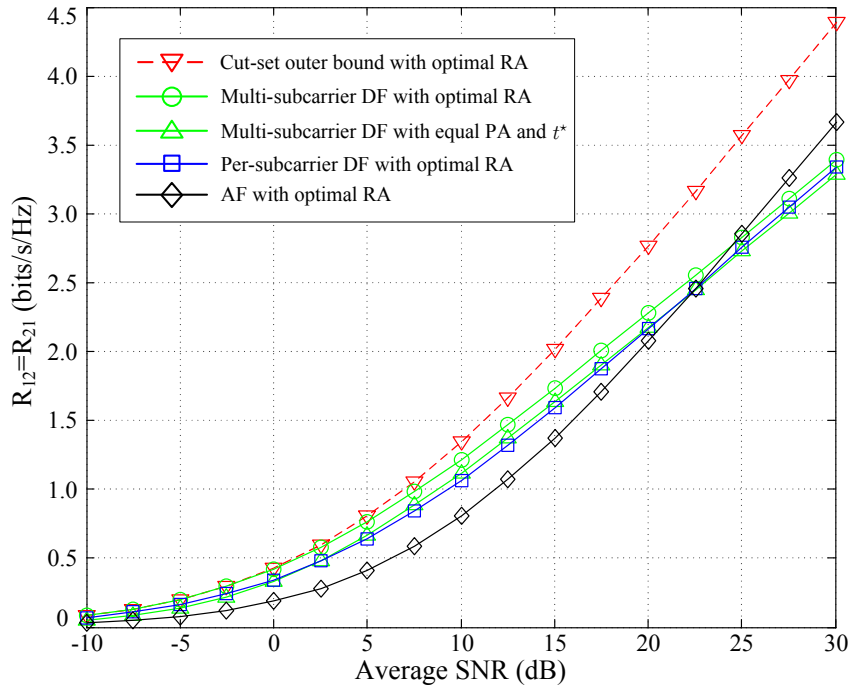}}
    \caption{Achievable rate performance comparison of two-way OFDM relay strategies, where $\text{SNR}_1\,{=}\,\text{SNR}_2\,{=}\,\text{SNR}$, $R_{12}=R_{21}$, ``RA'' stands for resource allocation, and ``equal PA'' stands for equal power allocation in the legend.} \label{fig:rate-msc-psc}
\end{figure}
\else
\fi

\appendices
\section{Proof of Lemma~\ref{prop1}}\label{sec:proof-prop1}
According to Theorem~2 of \cite{Kim_TIT08}, the optimal achievable rate region of discrete memoryless two-way relay channel with a DF relay strategy is given by the set of rate pairs $(R_{12},R_{21})$ satisfying
\vspace{-18pt}
\begin{subequations}
\begin{align}
&R_{12}\leq\min\!\left\{tI(X_1;Y_R|X_2),(1-t)I(X_R;Y_2)\right\},\label{eq:ineq1}\\[-8pt]
&R_{21}\leq\min\!\left\{tI(X_2;Y_R|X_1),(1-t)I(X_R;Y_1)\right\},\label{eq:ineq2}\\[-8pt]
&R_{12}+R_{21}\leq tI(X_1,X_2;Y_R),\label{eq:ineq3}\\[-35pt]\nonumber
\end{align}
\end{subequations}
where $X_i$ and $Y_i$ $(i{=}1,2,R)$ are the input and output symbols of the channel at the terminal and relay nodes, respectively.

In the two-way parallel Gaussian relay channel, the channel input and output symbols are given by the vectors $X_i=(X_{i1},\ldots,X_{iN})$ and $Y_i=(Y_{i1},\ldots,Y_{iN})$, respectively. The mutual information terms in \eqref{eq:ineq1}-\eqref{eq:ineq3} can be maximized simultaneously with the following channel input distributions \cite[Section~9.4]{Cover06}:
\begin{enumerate}
\item The elements of channel input $X_{in}$ should be statistically independent for different $n$;
\item The elements of channel input $X_{in}$ should be Gaussian random variables with zero mean and unit variance.
\end{enumerate}
By applying these channel input distributions, and by further considering the power and channel resource constraints, the achievable rate region \eqref{eq:prop1} is attained.

\section{Proof of Proposition~\ref{prop2}}\label{sec:proof-prop2}
The optimal primal variables $(\bm{p}_1^{\star},\bm{p}_2^{\star})$ and the optimal dual variables $(\bm{{\lambda}}^{\star},\bm{{\alpha}}^{\star})$ to problem \eqref{eq:MA-problem} must satisfy the KKT condition
\vspace{-9pt}
\begin{equation} \label{eq:MA-KKT-rate}
\frac{\partial{L_{\text{MA}}}}{\partial{R_{\text{MA}}}}=\lambda_1^{\star}+\lambda_2^{\star}+\lambda_3^{\star}-1=0,
\vspace{-9pt}
\end{equation}
and the complementary slackness conditions
\vspace{-9pt}
\begin{equation} \label{eq:MA-KKT-slackness}
\lambda_k^{\star}\big[R_\text{MA}^{\star}-r_k(\bm{p}_1^{\star},\bm{p}_2^{\star})\big]=0,\;k=1,2,3.
\vspace{-9pt}
\end{equation}
By \eqref{eq:MA-KKT-rate}, the optimal dual variable $\bm{\lambda}^{\star}$ has at most two independent variables, i.e., $\lambda_1^{\star} = 1-\lambda_2^{\star}-\lambda_3^{\star}$. For convenience, $r_k(\bm{p}_1^{\star},\bm{p}_2^{\star})$ is simply denoted as $r_k^{\star}$ for $k=1,2,3$. If $r_1^{\star}\neq r_2^{\star}$, then by the complementary slackness conditions in \eqref{eq:MA-KKT-slackness}, the optimal dual variable $\bm{\lambda}^{\star}$ must satisfy
either $\bm{\lambda}^{\star}=(1{-}\lambda_3^{\star},0,\lambda_3^{\star})$ with $\lambda_2^{\star}=0$ or $\bm{\lambda}^{\star}=(0,1{-}\lambda_3^{\star},\lambda_3^{\star})$ with $\lambda_1^{\star}=0$, and the asserted statement is thus proved. Therefore, we only need to consider the case of $r_1^{\star}= r_2^{\star}$.

It can be easily seen from \eqref{eq:rate-fun1}-\eqref{eq:rate-fun3} that
\vspace{-4pt}
\begin{equation} \label{eq:MA-rate-ineq}
r_1^{\star}+{\rho}r_2^{\star}=
t\sum_{n=1}^N{\log_2\!\Big(1+\frac{g_{1n}p_{1n}^{\star}+g_{2n}p_{2n}^{\star}}t
+\frac{g_{1n}g_{2n}p_{1n}^{\star}p_{2n}^{\star}}{t^2}\Big)}\geq (\rho+1)r_3^{\star},
\end{equation}
and the equality holds in \eqref{eq:MA-rate-ineq} if and only if $p_{1n}^{\star}p_{2n}^{\star}=0$ for $n=1,\ldots,N$. This leads to two cases to be discussed as follows:

\textbf{Case~1}: $r_1^{\star}+{\rho}r_2^{\star}>(\rho+1)r_3^{\star}$.

Since $r_1^{\star}= r_2^{\star}$, we have that $r_3^{\star}<\frac{1}{\rho+1}r_1^{\star}+\frac{\rho}{\rho+1}r_2^{\star}=r_1^{\star}= r_2^{\star}$. If $\lambda_1^{\star}>0$ and $\lambda_2^{\star}>0$, then the complementary slackness conditions in \eqref{eq:MA-KKT-slackness} imply $R_\text{MA}^{\star}=r_1^{\star}=r_2^{\star}>r_3^{\star}$, which contradicts with the rate constraint $R_\text{MA}^{\star}\,{\leq}\,r_3^{\star}$. Therefore, $\lambda_1^{\star}$ and $\lambda_2^{\star}$ can not be both positive, and Proposition~\ref{prop2} is thus proved in Case~1.

\textbf{Case~2}: $r_1^{\star}+{\rho}r_2^{\star}=(\rho+1)r_3^{\star}$.

Since $r_1^{\star}= r_2^{\star}$, we have $R_{\text{MA}}^\star=r_1^{\star}= r_2^{\star}=r_3^{\star}$. If problem \eqref{eq:dual-problem} has an optimal dual solution $(\bm{{\lambda}}^{\star},\bm{{\alpha}}^{\star})$ with ${\lambda}_1^{\star}{\lambda}_2^{\star}=0$, the optimal dual variable $\bm{\lambda}^{\star}$ already satisfies either
$\bm{\lambda}^{\star}=(1{-}\lambda_3^{\star},0,\lambda_3^{\star})$ or $\bm{\lambda}^{\star}=(0,1{-}\lambda_3^{\star},\lambda_3^{\star})$. Suppose that there is an optimal dual point $\bm{\lambda}^{\star}$ satisfying ${\lambda}_1^{\star}>0$ and ${\lambda}_2^{\star}>0$, we will construct another optimal dual solution with the desired structure stated in Proposition~\ref{prop2}.

By \eqref{eq:MA-rate-ineq}, $r_1^{\star}+{\rho}r_2^{\star}=(\rho+1)r_3^{\star}$ happens only if the optimal primal solution satisfies $p_{1n}^{\star}p_{2n}^{\star}=0$ for all $n$. Let us define $\Iset_1\subseteq\Nset\triangleq\{1,\ldots,N\}$ as the index set of subcarriers with $p_{1n}^{\star}\geq0,p_{2n}^{\star}=0$, and $\Iset_2\subseteq\Nset$ with $p_{1n}^{\star}=0,p_{2n}^{\star}\geq0$.
The optimal primal variables $(\bm{p}_1^{\star},\bm{p}_2^{\star})$ and the optimal dual variables $(\bm{{\lambda}}^{\star},\bm{{\alpha}}^{\star})$ to problem \eqref{eq:MA-problem} must satisfy the following KKT conditions:
\vspace{-2pt}
\begin{subequations}
\begin{align}
&\frac{\partial{L_{\text{MA}}}}{\partial{R_{\text{MA}}}}=\lambda_1^{\star}+
\lambda_2^{\star}+\lambda_3^{\star}-1=0\,,\label{eq:MA-KKT-rate1}\\[-4pt]
&\frac{\partial{L_{\text{MA}}}}{\partial{p_{1n}}}=
\alpha_1^{\star}-\frac{t{g_{1n}}[\lambda_3^{\star}+(\rho+1)\lambda_1^{\star}]}
{(\rho+1)(t+g_{1n}p_{1n}^{\star})\ln2}
\;\left\{\begin{array}{ll}
\!\!\geq0,&\text{if}\;p_{1n}^{\star}=0 \\[2pt]
\!\!=0,&\text{if}\;p_{1n}^{\star}>0
\end{array}\right.\!\!\!,\;n\In\Iset_1,\\[-2pt]
&\frac{\partial{L_{\text{MA}}}}{\partial{p_{2n}}}=
\alpha_2^{\star}-\frac{t{g_{2n}}[\rho\lambda_3^{\star}+(\rho+1)\lambda_2^{\star}]}
{\rho(\rho+1)(t+g_{2n}p_{2n}^{\star})\ln2}
\;\left\{\begin{array}{ll}
\!\!\geq0,&\text{if}\;p_{2n}^{\star}=0 \\[2pt]
\!\!=0,&\text{if}\;p_{2n}^{\star}>0
\end{array}\right.\!\!\!,\;n\In\Iset_2,\\[-5pt]
&\lambda_k^{\star}\geq0,\;k=1,2,3,\\[-3pt]
&R_\text{MA}^{\star}-r_k^{\star}\leq0\,,\;k=1,2,3,\\[-3pt]
&\lambda_k^{\star}\big(R_\text{MA}^{\star}-r_k^{\star}\big)=0\,,\;k=1,2,3,\\[-3pt]
&\alpha_i^{\star}\geq0,\;i=1,2,\\[-3pt]
&{\textstyle \sum_{n=1}^N{p_{in}^{\star}}}-P_i\leq0\,,\;i=1,2,\label{eq:MA-KKT-feasible1}\\[-3pt]
&\alpha_i^{\star}\big({\textstyle \sum_{n=1}^N{p_{in}^{\star}}}-P_i\big)=0\,,\;i=1,2.\label{eq:MA-KKT-slackness1}\\[-35pt]\nonumber
\end{align}
\end{subequations}

If $\lambda_1^{\star}\,{\geq}\frac{\,1\,}{\rho}\lambda_2^{\star}\,{>}\,0$, we define a new dual point $\widetilde{\bm{\lambda}}=\big(\lambda_1^{\star}{-}\frac{\,1\,}{\rho}
\lambda_2^{\star},0,\lambda_3^{\star}{+}\frac{\rho+1}{\rho}\lambda_2^{\star}\big)$. Since $\widetilde{{\lambda}}_1{+}\widetilde{{\lambda}}_2{+}\widetilde{{\lambda}}_3=\lambda_1^{\star}{+}
\lambda_2^{\star}{+}\lambda_3^{\star}$, $\widetilde{{\lambda}}_3{+}(\rho{+}1)\widetilde{{\lambda}}_1=\lambda_3^{\star}{+}(\rho{+}1)\lambda_1^{\star}$, $\rho\widetilde{{\lambda}}_3{+}(\rho{+}1)\widetilde{{\lambda}}_2=\rho\lambda_3^{\star}{+}(\rho{+}1)\lambda_2^{\star}$, and $R_{\text{MA}}^\star=r_1^{\star}= r_2^{\star}=r_3^{\star}$,
the dual point $(\bm{\widetilde{\lambda}},\bm{{\alpha}}^{\star})$ and the primal point $(\bm{p_1}^{\star},\bm{p_2}^{\star})$ also satisfy the KKT conditions \eqref{eq:MA-KKT-rate1}-\eqref{eq:MA-KKT-slackness1}. Therefore, $\bm{\widetilde{\lambda}}$ is an optimal dual solution of problem \eqref{eq:dual-problem} that satisfies
$\bm{\widetilde{\lambda}}=(1{-}\widetilde{\lambda}_3,0,\widetilde{\lambda}_3)$.

If $0\,{<}\,\lambda_1^{\star}\,{<}\frac{\,1\,}{\rho}\lambda_2^{\star}$, similarly we can define another dual point $\hat{\bm{\lambda}}=\big(0,\lambda_2^{\star}{-}\rho\lambda_1^{\star},\lambda_3^{\star}{+}(\rho{+}1)\lambda_1^{\star}\big)$.
Since $\hat{{\lambda}}_1{+}\hat{{\lambda}}_2{+}\hat{{\lambda}}_3=\lambda_1^{\star}{+}
\lambda_2^{\star}{+}\lambda_3^{\star}$, $\hat{{\lambda}}_3+(\rho+1)\hat{{\lambda}}_1=\lambda_3^{\star}+(\rho+1)\lambda_1^{\star}$,
$\rho\hat{{\lambda}}_3+(\rho+1)\hat{{\lambda}}_2=\rho\lambda_3^{\star}+(\rho+1)\lambda_2^{\star}$, and $R_{\text{MA}}^\star=r_1^{\star}= r_2^{\star}=r_3^{\star}$, the dual point $(\bm{\hat{\lambda}},\bm{{\alpha}}^{\star})$ and the primal point $(\bm{p_1}^{\star},\bm{p_2}^{\star})$ also satisfy the KKT conditions \eqref{eq:MA-KKT-rate1}-\eqref{eq:MA-KKT-slackness1}. Therefore, $\bm{\hat{\lambda}}$ is an optimal dual solution of problem \eqref{eq:dual-problem} that satisfies
$\bm{\hat{\lambda}}=(0,1{-}\hat{\lambda}_3,\hat{\lambda}_3)$. Hence, the statement of Proposition~\ref{prop2} has been proved for Case~2.

\section{Closed-form Solution to \eqref{eq:MA-KKT} without Using Proposition~\ref{prop2}\\ (Discussed in Remark 1)} \label{sec:cubic-eq}
When the structural property in Proposition~\ref{prop2} is not available, the primal power allocation solution is more complicated for the case of  $p_{1n}^{\star}>0,\,p_{2n}^{\star}>0$. In this case, the KKT conditions \eqref{eq:MA-KKT1} and \eqref{eq:MA-KKT2} both hold with equality. Therefore, we need to solve a system of quadratic equations with two variables. To simplify this problem, we define an auxiliary variable
\vspace{-9pt}
\begin{equation} \label{eq:cubic-aux}
x\triangleq g_{1n}p_{1n}^{\star}+g_{2n}p_{2n}^{\star}.
\vspace{-10pt}
\end{equation}
Then, by \eqref{eq:MA-KKT} and through some derivations, we obtain
\vspace{-4pt}
\begin{subequations} \label{eq:MA-power}
\begin{align} \label{eq:MA-power1}
p_{1n}^{\star}&=~\frac{t(\rho+1)\lambda_1}{\alpha_1(\rho+1)\ln2{-}t g_{1n}\lambda_3/(t+x)}-\frac{t}{g_{1n}},\\[-4pt]
p_{2n}^{\star}&=~\frac{t(\rho+1)\lambda_2/\rho}{\alpha_2(\rho+1)\ln2{-}t g_{2n}\lambda_3/(t+x)}-\frac{t}{g_{2n}}. \label{eq:MA-power2}
\vspace{-10pt}
\end{align}
\end{subequations}

By substituting \eqref{eq:MA-power1} and \eqref{eq:MA-power2} into \eqref{eq:cubic-aux}, we end up with the following \emph{cubic} equation of $x$:
\vspace{-2pt}
\begin{equation} \label{eq:cubic-eq}
\frac{tg_{1n}(\rho+1)\lambda_1}{\alpha_1(\rho+1)\ln2{-}tg_{1n}\lambda_3/(t+x)}
+\frac{tg_{2n}(\rho+1)\lambda_2/\rho}{\alpha_2(\rho+1)\ln2{-}tg_{2n}\lambda_3/(t+x)}
=x+2t.
\end{equation}

It is widely known that the closed-form solutions of the \emph{cubic} equation $x^3+ax^2+bx+c=0$ are given by Cardano's formula \cite{Dunham90}, i.e.,
\vspace{-4pt}
\begin{subequations} \label{eq:cubic-eq-sol}
\begin{align}
&x_1=e^{j\theta_1}\sqrt[3]{\big|-q/2+\sqrt{\Delta}\big|}
+e^{j\theta_2}\sqrt[3]{\big|-q/2-\sqrt{\Delta}\big|}-a/3,\\[-4pt]
&x_2=\omega e^{j\theta_1/3}\sqrt[3]{\big|-q/2+\sqrt{\Delta}\big|}
+\omega^2e^{j\theta_2/3}\sqrt[3]{\big|-q/2-\sqrt{\Delta}\big|}-a/3,\\[-4pt]
&x_3=\omega^2e^{j\theta_1/3}\sqrt[3]{\big|-q/2+\sqrt{\Delta}\big|}
+\omega e^{j\theta_2/3}\sqrt[3]{\big|-q/2-\sqrt{\Delta}\big|}-a/3,\\[-35pt]\nonumber
\end{align}
\end{subequations}
where $p=-a^2/3+b,\,q=2a^3/27-ab/3+c,\,\omega=-1/2+j\sqrt{3}/2,\,\Delta=p^3/27+q^2/4,\,
\theta_1=\mathrm{angle}\,(-q/2+\sqrt{\Delta}),\,\theta_2=\mathrm{angle}\,(-q/2-\sqrt{\Delta})$ , and $\mathrm{angle}\,(\cdot)$ denotes the phase angle of an complex number.
If $\Delta\geq0$, the cubic equation has one real root and a pair of conjugate complex roots; if $\Delta<0$, the cubic equation has three real roots.

After obtaining the positive real root $x$ of \eqref{eq:cubic-eq}, we can easily obtain the optimal $p_{1n}^{\star}$ and $p_{2n}^{\star}$ by substituting $x$ into \eqref{eq:MA-power}, which is the closed-form power allocation solution.

\ifreport
\section{Bounds for Power Dual Variables}\label{sec:alpha-bounds}
The optimal $\bm{\alpha}^{\star}(\bm{\lambda})=(\alpha_1^{\star},\alpha_2^{\star})$ for problem \eqref{eq:alpha-max} must satisfy the KKT conditions \eqref{eq:MA-KKT}, and there must exist $n_1\,(1{\leq}n_1{\leq}N)$ such that $p_{1n_1}>0$. Thus, by the feasible conditions $ p_{in}^{\star}\geq0$ $(i{=}1,2,\,n{=}1,\ldots,N)$ and \eqref{eq:MA-KKT1} with $n=n_1$, we can obtain an upper bound for $\alpha_1^{\star}$,
{\setlength\arraycolsep{1pt} \begin{eqnarray}
\alpha_1^{\star}&&=\frac{t{g_{1n_1}}\lambda_3/(\rho+1)}{(t+g_{1n_1}p_{1n_1}^{\star}+g_{2n_1}p_{2n_1}^{\star})\ln2}
+\frac{t{g_{1n_1}}\lambda_1}{(t+g_{1n_1}p_{1n_1}^{\star})\ln2} \nonumber\\
&&\leq\frac{{g_{1n_1}}[(\rho+1)\lambda_1+\lambda_3]}{(\rho+1)\ln2}
\leq\frac{(\rho+1)\lambda_1+\lambda_3}{(\rho+1)\ln2}\max_n\,\{g_{1n}\},\nonumber
\end{eqnarray}}
and a trivial lower bound $\alpha_1^{\star}\geq0$.

Similar discussions can be applied for $\alpha_2^{\star}$ and the optimal $\alpha_3^{\star}(\lambda_5)$ for problem \eqref{eq:dual-problem-BC}, and thus we can attain that
\vspace{-9pt}
\begin{subequations}
\begin{align}
&0=\alpha_{1,\min}\leq\alpha_1^{\star}\leq\alpha_{1,\max}=\frac{(\rho+1)\lambda_1+\lambda_3}{(\rho+1)\ln2}\max_n\,\{g_{1n}\},\label{eq:alpha1-bounds}\\
&0=\alpha_{2,\min}\leq\alpha_2^{\star}\leq\alpha_{2,\max}=\frac{(\rho+1)\lambda_2+\lambda_3}{(\rho+1)\ln2}\max_n\,\{g_{2n}\},\label{eq:alpha2-bounds}\\
&0=\alpha_{3,\min}\leq\alpha_3^{\star}\leq\alpha_{3,\max}=\max_n\!\left\{\frac{\rho\tilde{g}_{2n}(1{-}\lambda_5){+}\tilde{g}_{1n}\lambda_5}{\rho\ln2}\right\}.\label{eq:alpha3-bounds}
\end{align}
\end{subequations}
By \eqref{eq:alpha1-bounds} and \eqref{eq:alpha2-bounds}, we can initialize $\bm{\alpha}$ and the matrix $\mathbf{A}$ in Algorithm~\ref{alg1} as
\begin{eqnarray}
\bm{\alpha}=(\alpha_{1,\max}/2,\alpha_{1,\max}/2),\;\mathbf{A}=\left[\begin{array}{cc}
\alpha_{1,\max}^2/4 & 0 \\
0 & \alpha_{2,\max}^2/4
\end{array}\right],
\end{eqnarray}
and the initialization of $\alpha_{3,\min}$ and $\alpha_{3,\max}$ in Algorithm~\ref{alg3} is given by \eqref{eq:alpha3-bounds}.
\else
\fi

\section{Proof of Proposition~\ref{prop3}}\label{sec:proof-prop3}
By \eqref{eq:MA-rate-ineq}, it is known that $r_1+\rho r_2\,{\geq}\,(\rho+1)r_3$. This leads to two cases to be discussed:

\textbf{Case~1}: $r_1+\rho r_2\,{>}\,(\rho+1)r_3$.

In this case, if $r_3\geq r_1$, we have $\rho(r_2-r_3)\,{>}\,r_3-r_1\geq0$. Hence, $r_2>r_3$. Assume $\bm{\lambda}^{\star}=(0,1{-}\lambda_3^{\star},\lambda_3^{\star})\In\bm{\Lambda}_2{\setminus}\{\bm{\lambda}^0\}$, which means $0\leq\lambda_3^{\star}<1$, and then we have
\vspace{-9pt}
\begin{equation} \label{eq:prop3-case1}
(1-\lambda_3^{\star})(r_3-r_2)<0.
\vspace{-9pt}
\end{equation}
On the other hand, since $\bm{\lambda}^{\star}=(0,1{-}\lambda_3^{\star},\lambda_3^{\star})$ is an optimal dual point, by \eqref{eq:subgradient-lambda} and \eqref{eq:subgradient-lambda-optimal}, it must be true that
\vspace{-9pt}
\begin{equation}
(\bm{\lambda}^{\star}-\bm{\lambda}^0)^T\bm{\xi}(\bm{\lambda}^0)=(1-\lambda_3^{\star})(r_3-r_2)\geq0,
\vspace{-9pt}
\end{equation}
which leads to a contradiction with \eqref{eq:prop3-case1}. Thus, $\bm{\lambda}^{\star}\notin\bm{\Lambda}_2{\setminus}\{\bm{\lambda}^0\}$. By Proposition~\ref{prop2}, we must have $\bm{\lambda}^{\star}\In\bm{\Lambda}_1$.

Similarly, if $r_3\geq r_2$, we can show that $\bm{\lambda}^{\star}\In\bm{\Lambda}_2$.

\textbf{Case 2}:  $r_1+\rho r_2=(\rho+1)r_3$.

If only one of the inequalities of $r_3\geq r_1$ and $r_3\geq r_2$ is satisfied, similar to Case~1, we can show that $\bm{\lambda}^{\star}\In\bm{\Lambda}_1$ if $r_3\geq r_1$ and $\bm{\lambda}^{\star}\In\bm{\Lambda}_2$ if $r_3\geq r_2$.

If both $r_3\geq r_1$ and $r_3\geq r_2$, we have $r_1=r_2=r_3=R_{\text{MA}}^{\star}$ by the condition of Case~2. Thus, according to \eqref{eq:subgradient-lambda}, the subgradient $\bm{\xi}(\bm{\lambda}^0)=\bm{0}$. Substituting this into \eqref{eq:subgradient-lambda-def} yields
\vspace{-9pt}
\begin{equation}
G_\text{MA}(\bm{\lambda})\leq G_\text{MA}(\bm{\lambda}^0),\quad\forall~ \bm{\lambda}\In\bm{\Lambda}_1\,{\textstyle\bigcup}\,\bm{\Lambda}_2,
\vspace{-9pt}
\end{equation}
which means that $\bm{\lambda}^0$ itself is an optimal solution to \eqref{eq:lambda-max1}. i.e., $\bm{\lambda}^{\star}=\bm{\lambda}^0$.

\ifreport
\section{Proof of Proposition~\ref{prop5}}\label{sec:proof-prop5}
For sufficiently small $x$, since $p_{in}\leq P_i=x\bar{P}_i$, $p_{in}$ is of the order  $O(x^b)$ for $b\geq1$, and $p_{in}^2$ is of the order $O(x^{2b})$.
By this and \eqref{eq:taylor-low}, each rate pair $(R_{12},R_{21})\In\mathcal{R}_\text{DF}(x\bar{\bm P},\mathcal{G})$ satisfies
\vspace{-4pt}
\begin{subequations}
\begin{align}
&R_{12}\leq \frac{1}{\ln2} \min\bigg\{\sum_{n=1}^Ng_{1n}p_{1n},\sum_{n=1}^N\tilde{g}_{2n}p_{Rn}\bigg\}+O(x^{2b}),\label{eq:lowSNR1}\\[-4pt]
&R_{21}\leq \frac{1}{\ln2}  \min\bigg\{\sum_{n=1}^Ng_{2n}p_{2n},\sum_{n=1}^N\tilde{g}_{1n}p_{Rn}\bigg\}+O(x^{2b}),\label{eq:lowSNR2}\\[-4pt]
&R_{12}+R_{21}\leq \frac{1}{\ln2}\sum_{n=1}^N\left( g_{1n}p_{1n}+g_{2n}p_{2n}\right)+O(x^{2b}).\label{eq:lowSNR3}
\end{align}
\end{subequations}
By taking the summation of \eqref{eq:lowSNR1} and \eqref{eq:lowSNR2}, it is easy to see that the sum-rate constraint \eqref{eq:lowSNR3} always holds if \eqref{eq:lowSNR1} and \eqref{eq:lowSNR2} are true. In other words, the achievable rate region $\mathcal{R}_\text{DF}(x\bar{\bm P},\mathcal{G})$ can be expressed by \eqref{eq:lowSNR1} and \eqref{eq:lowSNR2} for sufficiently small $x$. On the other hand, the cut-set outer bound region $\mathcal{R}_\text{out}(x\bar{\bm P},\mathcal{G})$ is also described by \eqref{eq:lowSNR1} and \eqref{eq:lowSNR2}. Taking resource allocation into account, both $\mathcal{R}_\text{DF}(x\bar{\bm P},\mathcal{G})$ and $\mathcal{R}_\text{out}(x\bar{\bm P},\mathcal{G})$ are in the form of
{\setlength\arraycolsep{0pt}
\begin{eqnarray}
\mathcal{R}(x\bar{\bm P},\mathcal{G})=\Bigg\{
&&(R_{12},R_{21})\In\Rset_+^2~\bigg|\nonumber\\[-4pt]
&&R_{12}\leq \frac{1}{\ln2}\min\bigg\{\sum_{n=1}^Ng_{1n}p_{1n},\sum_{n=1}^N\tilde{g}_{2n}p_{Rn}\bigg\}+O(x^{2b}),\nonumber\\[-4pt]
&&R_{21}\leq \frac{1}{\ln2}  \min\bigg\{\sum_{n=1}^Ng_{2n}p_{2n},\sum_{n=1}^N\tilde{g}_{1n}p_{Rn}\bigg\}+O(x^{2b}),\nonumber\\[-4pt]
&&\sum_{n=1}^N{p_{in}}\leq P_i,~p_{in}\geq0,\,i=1,2,R,\,n=1,\ldots,N\Bigg\},
\end{eqnarray}}
$\!\!$for sufficiently small $x$. Therefore, for any rate pair $(R_{12},R_{21})\In\mathcal{R}_\text{out}(x\bar{\bm P},\mathcal{G})$, there exists some $(\hat{R}_{12}, \hat{R}_{21})\In\mathcal{R}_\text{DF}(x\bar{\bm P},\mathcal{G})$ such that $R_{12}=\hat{R}_{12}+O(x^{2b})$ and $R_{21}=\hat{R}_{21}+O(x^{2b})$ for $b\geq1$. Now let $x\,{\rightarrow}\,0$, the rate regions $\mathcal{R}_\text{DF}(x\bar{\bm P},\mathcal{G})$ and $\mathcal{R}_\text{out}(x\bar{\bm P},\mathcal{G})$ tend to be the same.
\else
\fi

\ifreport
\section{Proof of Proposition~\ref{prop7}}\label{sec:proof-prop7}
First, we consider an equal power allocation scheme $p_{in}={P}_i/N=x\bar{P}_i/N$, which leads to a rate region given by
\vspace{-6pt}
{\setlength\arraycolsep{1pt}
\begin{eqnarray}\label{eq:proof-innerregion}
\mathcal{R}_1(x\bar{\bm P},\mathcal{G})=\Bigg\{
(R_{12},R_{21})\In\Rset_+^2~\bigg|
&&R_{12}\leq\sum_{n=1}^N{t}\log_2\!\Big(1\!+\!\frac{g_{1n}x\bar{P}_1}{tN}\Big),\nonumber\\[-4pt]
&&R_{21}\leq\sum_{n=1}^N{t}\log_2\!\Big(1\!+\!\frac{g_{2n}x\bar{P}_2}{tN}\Big),\nonumber\\[-4pt]
&&R_{12}+R_{21}\leq\sum_{n=1}^Nt{\log_2\!\Big(1\!+\!\frac{g_{1n}x\bar{P}_1+g_{2n}x\bar{P}_2}{tN}\Big)},\nonumber\\[-4pt]
&&R_{12}\leq\sum_{n=1}^N(1-{t})\log_2\!\Big(1\!+\!\frac{\tilde{g}_{2n}x\bar{P}_R}{(1-t)N}\Big),\nonumber\\[-4pt]
&&R_{21}\leq\sum_{n=1}^N(1-{t})\log_2\!\Big(1\!+\!\frac{\tilde{g}_{1n}x\bar{P}_R}{(1-t)N}\Big),\,0<t<1\Bigg\}.
\end{eqnarray}}
$\!\!$Since this equal power allocation scheme is feasible, the rate region \eqref{eq:proof-innerregion} satisfies
\vspace{-4pt}
\begin{eqnarray}\label{eq:con1}
\mathcal{R}_1(x\bar{\bm P},\mathcal{G})\subset \mathcal{R}_\text{DF}(x\bar{\bm P},\mathcal{G}).
\end{eqnarray}
$\!\!$All the rate functions in \eqref{eq:proof-innerregion} are of the form $\log_2(1+ax)$ for $a>0$. According to \eqref{eq:taylor-high}, it can be expressed by $\log_2(x)+O(1)$ for sufficiently large $x$. By this, for sufficiently large $x$, we can obtain that
\vspace{-4pt}
{\setlength\arraycolsep{1pt}
\begin{eqnarray}
\mathcal{R}_1(x\bar{\bm P},\mathcal{G})=\Bigg\{
(R_{12},R_{21})\In\Rset_+^2~\bigg|
&&R_{12}\leq\sum_{n=1}^N{t}\log_2(x)+O(1),\nonumber\\[-4pt]
&&R_{21}\leq\sum_{n=1}^N{t}\log_2(x)+O(1),\nonumber\\[-4pt]
&&R_{12}+R_{21}\leq\sum_{n=1}^Nt\log_2(x)+O(1),\nonumber\\[-4pt]
&&R_{12}\leq\sum_{n=1}^N(1-{t})\log_2(x)+O(1),\nonumber\\[-4pt]
&&R_{21}\leq\sum_{n=1}^N(1-{t})\log_2(x)+O(1),\,0<t<1\Bigg\}.
\end{eqnarray}}
$\!\!$Hence, similar to the definition \eqref{eq:multiplexing-DF}, the multiplexing gain region of ${\mathcal{R}}_1(x\bar{\bm P},\mathcal{G})$ is given by
\vspace{-4pt}
{\setlength\arraycolsep{1pt}
\begin{eqnarray}\label{eq:multiplexing1}
\Big\{(r_{12},r_{21})~\Big|
&&r_{12}\leq tN,\,r_{21}\leq tN,\,r_{12}+r_{21}\leq tN,\nonumber\\
&&r_{12}\leq(1-t)N,\,r_{21}\leq(1-t)N,\,0<t<1,\,r_{12},r_{21}\geq0\Big\}.
\end{eqnarray}}
$\!\!$After some simple manipulations, \eqref{eq:multiplexing1} can be simplified as \eqref{eq:multiplexing}.

Then, we consider an infeasible power allocation scheme $p_{in}=P_i=x\bar{P}_i $, which results in the following rate region:
\vspace{-4pt}
{\setlength\arraycolsep{1pt}
\begin{eqnarray}\label{eq:proof-outerregion}
\mathcal{R}_2(x\bar{\bm P},\mathcal{G})=\Bigg\{
(R_{12},R_{21})\In\Rset_+^2~\bigg|
&&R_{12}\leq\sum_{n=1}^N{t}\log_2\!\Big(1\!+\!\frac{g_{1n}x\bar{P}_1}{t}\Big),\nonumber\\[-4pt]
&&R_{21}\leq\sum_{n=1}^N{t}\log_2\!\Big(1\!+\!\frac{g_{2n}x\bar{P}_2}{t}\Big),\nonumber\\[-4pt]
&&R_{12}+R_{21}\leq\sum_{n=1}^Nt{\log_2\!\Big(1\!+\!\frac{g_{1n}x\bar{P}_1+g_{2n}x\bar{P}_2}{t}\Big)},\nonumber\\[-4pt]
&&R_{12}\leq\sum_{n=1}^N(1-{t})\log_2\!\Big(1\!+\!\frac{\tilde{g}_{2n}x\bar{P}_R}{1-t}\Big),\nonumber\\[-4pt]
&&R_{21}\leq\sum_{n=1}^N(1-{t})\log_2\!\Big(1\!+\!\frac{\tilde{g}_{1n}x\bar{P}_R}{1-t}\Big),\,0<t<1\Bigg\}.
\end{eqnarray}}
$\!\!$It can be easily seen from \eqref{eq:prop1} that the rate region \eqref{eq:proof-outerregion} satisfies
\vspace{-4pt}
\begin{eqnarray}\label{eq:con2}
\mathcal{R}_\text{DF}(x\bar{\bm P},\mathcal{G})\subset {\mathcal{R}}_2(x\bar{\bm P},\mathcal{G}),
\end{eqnarray}
and the multiplexing gain region of ${\mathcal{R}}_2(x\bar{\bm P},\mathcal{G})$ is also given by \eqref{eq:multiplexing1}, and thus \eqref{eq:multiplexing}. Therefore, the asserted statement follows from \eqref{eq:con1} and \eqref{eq:con2}.
\else
\fi

\section*{Acknowledgment}
The authors would like to thank ArulMurugan Ambikapathi, Wei-Chiang Li, Kun-Yu Wang, and Tsung-Han Chan for their valuable suggestions during the preparation of this paper.

\bibliographystyle{IEEEtran}
\linespread{1.0}\selectfont
\bibliography{refs_WCOM}

\end{document}